\documentclass[11pt]{article}
\usepackage[margin=1in]{geometry}
\usepackage{graphicx} 

\usepackage{xspace}
\usepackage{xcolor}
\usepackage{amscd}
\usepackage{amsmath}
\usepackage{amssymb}
\usepackage{amstext}
\usepackage{amsthm}
\usepackage{bbold}
\usepackage{bm}
\usepackage{colonequals}
\usepackage{tikz}

\usepackage{hyperref}

\usepackage{algpseudocode}
\usepackage[ruled,vlined]{algorithm2e}

\newcommand{\sM}{\mathcal{M}}

\newcommand{\NN}{\mathbb{N}}
\newcommand{\PP}{\mathbb{P}}
\newcommand{\QQ}{\mathbb{Q}}

\DeclareSymbolFont{bbold}{U}{bbold}{m}{n}
\DeclareSymbolFontAlphabet{\mathbbold}{bbold}

\newcommand{\sgn}{\mathrm{sgn}}

\newcommand{\Var}{\mathrm{Var}}

\newcommand{\Aut}{\mathrm{Aut}}
\newcommand{\Iso}{\mathrm{Iso}}

\DeclareMathOperator*{\E}{\mathbb{E}}

\newcommand{\Rad}{\mathrm{Rad}}
\newcommand{\Unif}{\mathrm{Unif}}
\newcommand{\Ber}{\mathrm{Ber}}
\newcommand{\Bin}{\mathrm{Bin}}

\newtheorem{theorem}{Theorem}[section]
\newtheorem{remark}[theorem]{Remark}

\newtheorem{lemma}[theorem]{Lemma}
\newtheorem{question}[theorem]{Question}

\newtheorem{definition}[theorem]{Definition}

\newtheorem{proposition}[theorem]{Proposition}
\newtheorem{corollary}[theorem]{Corollary}

\title{Near Optimal Algorithms for Noisy $k$-XOR under Low-Degree Heuristic}
\author{Songtao Mao\thanks{\texttt{smao13@jhu.edu}, Department of Computer Science, Johns Hopkins University.}}
\date{}

\begin{document}
\maketitle
\begin{abstract}
Noisy $k$-XOR is a basic average-case inference problem in which one observes random noisy $k$-ary parity constraints and seeks to recover, or more weakly, detect, a hidden Boolean assignment. A central question is to characterize the tradeoff among sample complexity, noise level, and running time.

We give a recovery algorithm, and hence also a detection algorithm, for noisy $k$-XOR in the high-noise regime. For every parameter $D$, our algorithm runs in time $n^{D+O(1)}$ and succeeds whenever
$$
m \ge C_k \frac{n^{k/2}}{D^{\,k/2-1}\delta^2},
$$
where $C_k$ is an explicit constant depending only on $k$, and $\delta$ is the noise bias. Our result matches the best previously known time--sample tradeoff for detection, while simultaneously yielding recovery guarantees. In addition, the dependence on the noise bias $\delta$ is optimal up to constant factors, matching the information-theoretic scaling.

We also prove matching low-degree lower bounds. In particular, we show that the degree-$D$ low-degree likelihood ratio has bounded $L^2$-norm below the same threshold, up to the same factor $D^{k/2-1}$. Under the low-degree heuristic, this implies that our algorithm is near-optimal over a broad range of parameters.

Our approach combines a refined second-moment analysis with color coding and dynamic programming for structured hypergraph embedding statistics. These techniques may be of independent interest for other average-case inference problems.
\end{abstract}
\newpage
\tableofcontents
\newpage
\section{Introduction}
The study of noisy $k$-XOR, also known as sparse LPN, lies at the intersection of average-case complexity, learning theory, and cryptography. In this problem, one observes random parity queries of Hamming weight $k$, together with noisy labels, and the goal is either to distinguish the planted distribution from pure noise or to recover the hidden Boolean assignment. From the cryptographic perspective, sparse LPN has become a widely used hardness assumption, underlying a variety of constructions including public-key encryption, pseudorandom generators, secure computation protocols, and other cryptographic primitives \cite{alekhnovich2003more,applebaum2010public,applebaum2017secure,couteau2021silver,dao2023multi,ragavan2025indistinguishability,raghuraman2023expand}.

We work with the equivalent $\{\pm1\}$-encoding of bits, in which $0$ and $1$ are represented by $+1$ and $-1$, respectively, so that parity constraints are written multiplicatively rather than additively over $\mathbb F_2$. In noisy $k$-XOR, there is an unknown vector $x=(x_1,\dots,x_n)\in\{\pm1\}^n$, and each sample consists of a uniformly random $k$-set of coordinates $\alpha=\{i_1,\dots,i_k\}\subseteq[n]$, together with a noisy observation of the parity $x_\alpha:=\prod_{i\in\alpha}x_i\in\{\pm1\}$. Under the planted model, the observed label takes the form $z(\alpha)=x_\alpha \xi_\alpha$, where $\xi_\alpha\in\{\pm1\}$ is an independent noise variable with mean $\E[\xi_\alpha]=\delta$. Equivalently, $z(\alpha)$ agrees with the true parity $x_\alpha$ with probability $(1+\delta)/2$, and is flipped with probability $(1-\delta)/2$. Given $m$ such samples, the goal is either to recover the hidden assignment $x$ (up to global sign when $k$ is even), or to distinguish this planted model from a null model in which the observed labels are independent unbiased random signs. Thus, \emph{distinguishing} and \emph{recovery} are two closely related but fundamentally different tasks: the former asks only whether the samples contain any detectable planted signal, whereas the latter asks to reconstruct the hidden assignment itself, and is therefore typically the stronger objective.

Ignoring for the moment the dependence on the noise parameter, the regime $m \gtrsim n^{k/2}$ is algorithmically more accessible. Around this scale, birthday-paradox-type collisions begin to appear among the sampled $k$-sparse queries, or among suitable projected queries. Such collisions yield repeated or overlapping noisy observations that can be aggregated to amplify the planted signal \cite{applebaum2016cryptographic}. Consequently, the more interesting regime is $m<n^{k/2}$, where such collisions are too sparse to be exploited directly.

There is by now a substantial line of work on distinguishing and refutation algorithms. Early work gave witness-based refutation for dense random formulas \cite{feige2006witnesses}, while later work established general refutation criteria for random CSPs based on the failure of $t$-wise independence \cite{allen2015refute}. Subsequent works further developed strong-refutation tradeoffs below the $n^{k/2}$ threshold for random and smoothed CSPs, together with structural simplifications and sharpenings of this framework \cite{raghavendra2017strongly,guruswami2022algorithms,hsieh2023simple}.

On the recovery side, substantially less is known. A foundational result of Feldman, Perkins, and Vempala gives a polynomial-time recovery algorithm for planted CSPs, including planted XOR-type problems, at roughly the $\widetilde O(n^{k/2})$ sample scale, via a reduction to noisy XOR together with subsampled power iteration \cite{feldman2015subsampled}. A separate line of work in cryptography studies \emph{search-to-decision} reductions, which aim to convert distinguishing algorithms into recovery algorithms. These results show that distinguishing information can sometimes be leveraged to recover the hidden secret, but typically at a nontrivial cost: some reductions \cite{applebaum2010public,applebaum2012pseudorandom,bogdanov2025sample} incur a substantially larger blow-up in the number of samples. In contrast, the reduction of Bogdanov, Sabin, and Vasudevan \cite{bogdanov2019xor} comes at the cost of an exponentially small success probability.

A complementary line of work studies noisy $k$-XOR across a range of noise regimes. Chen, Shu, and Zhou \cite{chen2025algorithms} developed improved algorithms for sparse LPN and LSLPN using domain-reduction and elimination-based techniques, with particularly strong guarantees when the noise rate is small. More recently, Basu et al.~\cite{basu2025solving} obtained a new algorithm for the \emph{search} version of noisy $k$-XOR. Their result gives an $n^{O(D)}$-time algorithm under the condition
\[
m \;\ge\; 2^{O(k)}\, n\log n \cdot
\max\!\left\{
\delta^{-11}\left(\frac{n}{D}\right)^{k/2-1},\;
\delta^{-2}
\right\}.
\]
Thus, although the dependence on $n$ and $D$ matches the best known refutation-style tradeoff, the leading term in the high-noise regime still has a substantially worse-than-quadratic dependence on the bias parameter $\delta$.

It is therefore natural to ask whether one can improve the dependence of the sample complexity on the
bias parameter $\delta$ while still achieving the best known running-time tradeoff.
The dependence $\delta^{-2}$ is information-theoretically natural: each sample is a noisy parity bit,
equivalently one use of a binary symmetric channel with crossover probability $(1-\delta)/2$, and hence
contributes only $\Theta(\delta^2)$ bits of mutual information in the high-noise regime.
Consequently, even for recovering an $n$-bit hidden vector in the absence of computational constraints,
one needs $m = \Theta(n/\delta^2)$ samples.
Our conjectural picture is further inspired by the recent work of Bresler and Harbuzova
\cite{bresler2026average}, which builds average-case reductions between noisy $k$-XOR and Tensor PCA. For the spiked covariance / PCA problem, the corresponding threshold is naturally expressed as $m \asymp n\delta^{-2}$.
Motivated by this perspective, together with the preceding discussion, we propose the following analogous questions for noisy $k$-XOR.
\begin{question}\label{q:eps2}
Does there exist an efficient algorithm for noisy $k$-XOR in the high-noise regime that achieves
the conjectured optimal quadratic dependence on the bias parameter?
More concretely, can one obtain detection or recovery with
\[
m = \widetilde{\Theta}\!\left(n^{k/2}\delta^{-2}\right),
\]
or, more generally, with the $n^{O(D)}$-time tradeoff
\[
m=\widetilde{\Theta}\!\left(n\left(\frac{n}{D}\right)^{k/2-1}\delta^{-2}\right),
\]
without incurring a worse polynomial dependence on $\delta$?
\end{question}

In this work, we answer this question affirmatively. We give a recovery algorithm for noisy $k$-XOR in the high-noise regime with high success probability, and therefore also obtain a distinguisher. Our result substantially improves the dependence on the bias parameter over the previous best recovery algorithm of Basu et al.~\cite{basu2025solving}, achieving the information-theoretically optimal $\delta^{-2}$ scaling. Moreover, unlike \cite{basu2025solving}, our algorithm incurs no additional $\log n$ loss in the runtime--sample tradeoff. More generally, the tradeoff among running time, sparsity $k$, and the parameter $D$ matches or improves upon the best previous distinguishing results.

A natural follow-up question is whether the tradeoff between sample complexity and running time must inherently exhibit an information--computation gap.

\begin{question}
Is there a computational barrier to detecting noisy $k$-XOR in time $n^{O(D)}$ when the sample size satisfies
\[
m=o\!\left(n\left(\frac{n}{D}\right)^{k/2-1}\delta^{-2}\right)?
\]
\end{question}

We provide evidence from the low-degree polynomial method, a framework inspired by the sum-of-squares hierarchy that has emerged as a powerful approach for predicting computational thresholds in high-dimensional inference problems \cite{barak2019nearly,hopkins2017efficient,hopkins2017power,hopkins2018statistical}. Roughly speaking, for a testing problem between two distributions $\mathbb{P}$ and $\mathbb{Q}$, one studies whether there exists a polynomial $f$ of degree $D$ whose expectation under $\mathbb{P}$ is significantly larger than its typical fluctuation under $\mathbb{Q}$. A convenient proxy is the low-degree advantage
\[
\mathrm{Adv}(f)
\;:=\;
\frac{\mathbb{E}_{\mathbb{P}}[f]}{\sqrt{\mathbb{E}_{\mathbb{Q}}[f^2]}}.
\]
The low-degree heuristic predicts that, if this quantity remains bounded for all degree-$D$ polynomials, then under constant noise no algorithm running in time roughly $n^{\Theta(D)}$ (or $e^{\widetilde\Theta(D)}$, depending on the parameterization) should strongly distinguish $\PP$ from $\QQ$. Although there are known counterexamples showing that, in certain noisy models, bounded low-degree statistics do not by themselves imply computational hardness \cite{holmgren2020counterexamples,buhai2025quasi}, the low-degree framework is still widely regarded as a robust predictor for many high-dimensional inference problems.

In our setting, the low-degree perspective suggests that the observed information--computation tradeoff between the sample size $m$ and the running time is intrinsic. More precisely, we prove a matching low-degree lower bound for noisy $k$-XOR: under the standard low-degree heuristic, any algorithm running in time $n^D$ must use
\[
m = \Omega\left( \frac{n^{k/2}}{D^{\,k/2-1}\delta^2}\right)
\]
samples. This matches the algorithmic tradeoff in the principal scaling. Thus, our upper bound is not only information-theoretically optimal in its $\delta^{-2}$ dependence, but also matches the conjectured computational threshold predicted by the low-degree method.
\subsection{Our results}
We begin by formalizing the noisy $k$-XOR model and the two algorithmic tasks that will be studied throughout the paper, namely distinguishing and recovery.

\begin{definition}[Noisy $k$-XOR]
Fix $k\ge 2$, and let $x\in\{\pm1\}^n$ be an unknown planted assignment. In the noisy $k$-XOR model, we observe $m$ samples $(\alpha_1,z_1),\dots,(\alpha_m,z_m)$,
where each $\alpha_t\in\binom{[n]}{k}$ is drawn uniformly at random, and
\[
z_t=\Bigl(\prod_{i\in\alpha_t}x_i\Bigr)\xi_t \in \{\pm1\},
\]
with $\xi_t$ independent and
$\Pr[\xi_t=1]=\frac{1+\delta}{2}$,
$\Pr[\xi_t=-1]=\frac{1-\delta}{2}$.
We denote this planted distribution by $\PP_x$. The null distribution $\QQ$ is defined in the same way, except that $z_1,\dots,z_m$ are independent uniform random signs, independent of the sampled $k$-tuples.

A \emph{distinguishing algorithm} outputs $0$ or $1$ given the samples. Its success probability is
\[
\frac12\Pr_{\PP_x}[\mathcal A=1]+\frac12\Pr_{\QQ}[\mathcal A=0].
\]
We say that $\mathcal A$ \emph{strongly distinguishes} if this probability is $1-o(1)$, and \emph{weakly distinguishes} if it is at least $1/2+\varepsilon$ for some fixed constant $\varepsilon>0$.

A \emph{recovery algorithm} outputs an estimate $\widehat x\in\{\pm1\}^n$ from samples drawn from $\PP_x$. Its success probability is
\[
\Pr_{\PP_x}[\widehat x=x]
\qquad\text{if $k$ is odd,}
\]
and
\[
\Pr_{\PP_x}[\widehat x\in\{x,-x\}]
\qquad\text{if $k$ is even.}
\]
If this probability is $1-o(1)$, we say the algorithm achieves \emph{exact recovery}.
\end{definition}
We now state our main positive result. It gives a recovery algorithm for noisy $k$-XOR with a smooth tradeoff between running time and sample complexity. In particular, for every sufficiently large multiple $D$ of $k$, we obtain an $n^{D+O(1)}$-time algorithm whose required sample size decreases as $D$ grows.

\begin{theorem}[Informal version of Theorem~\ref{thm:full-recovery-formal}]
Fix $k\ge 3$. Suppose $D$ is a multiple of $k$ and satisfies $Ck \;\le\; D \;\le\; \frac{0.1\log n}{\log\log n}$,
for some absolute constant $C>0$. Then there is an algorithm running in time
$n^{D+O(1)}$
that recovers, and hence also distinguishes, the planted assignment in noisy $k$-XOR, provided the number of samples satisfies
\[
m
\;\ge\;
\frac{2^{k} k^{k/2}}{k!}
\cdot
\frac{n^{k/2}}{\delta^2\,D^{\,k/2-1}},
\]
with success probability at least $2/3$. Moreover, if $D=\omega(1)$, or \[
m
=
\omega\!\left(
\frac{2^{k} k^{k/2}}{k!}
\cdot
\frac{n^{k/2}}{\delta^2\,D^{\,k/2-1}}
\right),
\]
then the success probability improves to $1-o(1)$. For odd $k$, the algorithm outputs $x$ exactly; for even $k$, it outputs one of $x$ and $-x$.
\end{theorem}

To complement this algorithmic upper bound, we next state a matching low-degree lower bound. This shows that, at least from the low-degree perspective, the tradeoff achieved by our algorithm is essentially optimal.

For the low-degree analysis, it is convenient to work in the Bernoulli model, in which each $\alpha\in\binom{[n]}{k}$ is observed independently, and when it is observed it receives a label $y_\alpha\in\{\pm1\}$, while unobserved coordinates are set to $0$. Thus an instance is encoded by a vector
\[
y\in\{0,\pm1\}^{\binom{[n]}{k}},
\]
and a polynomial test means a real polynomial in the coordinates $\{y_\alpha\}_{\alpha\in\binom{[n]}{k}}$, with degree taken with respect to these coordinates. For the low-degree lower bound, we let $\PP$ denote the mixture planted distribution obtained by averaging $\PP_x$ over a uniformly random hidden assignment $x\in\{\pm1\}^n$. By a standard symmetry reduction, this average-case formulation is equivalent to the fixed-$x$ formulation.

\begin{theorem}[Informal version of Theorem~\ref{thm:lowdegree-bound}]
Fix $k\ge 3$, and let $D=D(n)$ be a degree parameter. If
\[
m
=
o\!\left(
\frac{1}{e^{k/2+2}k^{k/2}}
\cdot
\frac{n^{k/2}}{D^{\,k/2-1}\delta^2}
\right),
\]
then no degree-$D$ polynomial test can distinguish $\PP$ from $\QQ$ with success probability $1/2+\varepsilon$, for any fixed constant $\varepsilon>0$.

Moreover, if $D$ is constant and
\[
m
=
O\!\left(
\frac{1}{e^{k/2+2}k^{k/2}}
\cdot
\frac{n^{k/2}}{D^{\,k/2-1}\delta^2}
\right),
\]
then no degree-$D$ polynomial test can distinguish $\PP$ from $\QQ$ with success probability $1-o(1)$.
\end{theorem}

In particular, under the low-degree heuristic, no degree-$D$ polynomial-time algorithm should be able to even weakly or strongly distinguish $\PP$ from $\QQ$ in these regimes.

Taken together, the two theorems indicate that our positive result is remarkably tight. Our upper bound matches the low-degree lower bound in the principal scaling
\[
\frac{n^{k/2}}{D^{\,k/2-1}\delta^2},
\]
and also nearly matches its dependence on $k$. In particular, when $k$ is superconstant, the leading factor $k^{-k/2}$ matches up to lower-order terms. This near-complete agreement gives strong evidence that our algorithm attains essentially the optimal runtime--sample tradeoff predicted by the low-degree framework.

We conclude by discussing an important caveat. In the low-noise regime, existing algorithms can go beyond what the standard constant-noise low-degree heuristic would suggest.

\paragraph{Limits of low-degree heuristic in the low-noise regime}
The algorithm of \cite{chen2025algorithms} achieves recovery in a parameter
regime beyond what the low-degree heuristic would suggest in the low-noise setting.

\begin{theorem}\cite[Theorem 2]{chen2025algorithms}
For any noisy $k$-XOR problem with constant $k$ and noise rate $\eta = (1-\delta)/2$, and for any $\varepsilon\in(0,1)$, there is a learning algorithm that, given
\[
m
:=
\max\left\{
1,\,
\frac{\log n}{k},\,
\frac{2\eta\,n^{\frac{1+\varepsilon}{2}}}{k^2}\log n
\right\}
\cdot
O(n)^{\,1+(1-\varepsilon)\cdot \frac{k-1}{2}}
\]
samples, returns the planted secret in time
\[
e^{\widetilde O\!\left(2\eta\,n^{\frac{1+\varepsilon}{2}}\right)}
\;+\;
n^{O(1)}\cdot m.
\]
\end{theorem}

For example, taking $\varepsilon=0.6$, $\eta=n^{-0.7}$, and $k=5$, the
theorem gives an algorithm using $m=O(n^{1.8})$
samples and running in time $e^{\widetilde O(n^{0.1})}$. This lies beyond the range suggested by standard low-degree methods. Nevertheless, it is not a standard counterexample to the low-degree conjecture, since the noise level in this regime, $\eta=n^{-0.7}$, vanishes with $n$ rather than remaining constant. Existing formulations of the low-degree heuristic are usually interpreted in the constant-noise setting, and to the best of our knowledge, there is no well-established conjectural framework that sharply captures the low-noise regime. This suggests an interesting direction for future work: formulating an appropriate low-degree conjecture when the error rate decays with $n$, and understanding the resulting interpolation between noise level and computational complexity.

\subsection{Related work}
\paragraph{Learning parity with noise}
The Learning Parity with Noise (LPN) problem is a central average-case problem in learning theory and cryptography \cite{blum1993cryptographic,blum2003noise,lyubashevsky2005parity,regev2009lattices,becker2012decoding,raz2018fast,bshouty2024approximating}. In its standard form, one is given random samples $(a,\langle a,s\rangle + e)$ over $\mathbb F_2$, where $a\in\mathbb F_2^n$ is uniformly random, $s\in\mathbb F_2^n$ is a hidden secret, and $e$ is a noise bit that flips the label with some probability. The noisy $k$-XOR problem can be viewed as a sparse variant of LPN, often called sparse LPN, in which each sample vector has Hamming weight $k$. Standard LPN is widely conjectured to be hard even in very low-noise regimes, for instance when the noise rate is as small as about $1/\sqrt n$, and the best known algorithms in this regime are still far from polynomial time, with running times of the form $e^{\Theta(\eta n)}$ when $\eta=o(1/\log n)$ and $2^{\Theta(n/\log n)}$ when $\eta=\Omega(1/\log n)$ \cite{blum2003noise}. It is further conjectured that standard LPN remains hard even when the noise rate is as small as $\eta=O(\log^2 n/n)$. Because of this apparent hardness, several easier variants of LPN have been studied. One important example is Learning Sparse Parities with Noise (LSPN) \cite{valiant2015finding,chen2025algorithms}, where the secret itself is $k$-sparse; this additional structure makes recovery more tractable. Another natural variant is batch LPN \cite{arora2011new,golowich2024exploration,li2025hardness}, where the learner receives samples in batches and the noise variables inside each batch may be correlated. More recently, \cite{bangachev2025near} established near-optimal reductions between running time and sparsity for noisy linear equations, although their results do not address sample complexity.

\paragraph{Solving random planted CSPs}
A line of work has studied the algorithmic task of solving random planted CSPs, where one seeks to recover a planted satisfying assignment from a highly random instance. In contrast to worst-case CSPs, where strong hardness results are known in sparse regimes while efficient algorithms exist only for very dense instances \cite{arora1995polynomial,haastad2001some,fotakis2015sub,impagliazzo2001complexity}, the random planted setting admits substantially better algorithms. Noisy $k$-XOR is a particularly structured special case of this framework, in which each constraint is a linear parity constraint of arity $k$ over $\mathbb F_2$, observed through independent label noise. Early and subsequent works developed algorithms and frameworks for planted CSPs and related average-case search problems \cite{barthel2002hiding,jia2007generating,bogdanov2009security,coja2010efficient,feldman2015subsampled}. In particular, \cite{feldman2015subsampled} gave a polynomial-time algorithm succeeding at about the $\widetilde O(n^{k/2})$ constraint scale, and more recent work has continued to refine the landscape, including connections to semirandom models and noisy $k$-XOR reductions \cite{abascal2021strongly,guruswami2022algorithms,hsieh2023simple,guruswami2023efficient,basu2025solving}.

\paragraph{Low-degree polynomial method.}
The low-degree method has become a central framework for predicting computational thresholds in high-dimensional inference. Its modern formulation was developed through the connection between polynomial tests and the sum-of-squares hierarchy \cite{hopkins2017efficient,hopkins2017power,hopkins2018statistical,kunisky2019notes}. For hypothesis testing and estimation problems, this viewpoint has led to sharp conjectural barriers in several major classes of models. For community-detection-type problems, see \cite{hopkins2017efficient,chen2026computational}. For sparse and correlated graph models, including dense subgraph detection, correlated Erd\H{o}s--R\'enyi, and graph lifts, see \cite{dhawan2025detection,ding2025low,kunisky2024computational}. For principal-component-type problems, including sparse PCA, constrained PCA, spiked-subspace recovery, and correlated spiked models, see \cite{bandeira2019computational,mao2021optimal,ding2024subexponential,bandeira2022franz,li2025algorithmic}. The same perspective has also been extended beyond pure detection to optimization and refutation problems, including random $k$-SAT, maximum independent set, and Boolean optimization landscapes \cite{bresler2022algorithmic,wein2022optimal,gamarnik2024hardness,wang2016average}. Taken together, these works suggest that bounded low-degree likelihood ratio is often the correct proxy for the onset of algorithmic hardness.
\section{Technical Overview}

In this section, we briefly outline the main ideas behind our upper and lower bounds.
\paragraph{Model selection and equivalence}
We consider three natural formulations of noisy $k$-XOR: the standard fixed-$m$ model with replacement, the fixed-$m$ model without replacement, and the Bernoulli-$p$ model. The fixed-$m$ formulations are standard in prior work and in the LPN literature. For the counting arguments and low-degree analysis in this paper, however, the Bernoulli model is more convenient, since the constraints are sampled independently. We therefore begin by showing that, after identifying the parameters via $p=\frac{m}{\binom{n}{k}}$,
the three models are asymptotically equivalent in the regime relevant to our results.

\paragraph{A general mean--variance framework}
Before introducing the specific hypergraph statistic, let us recall the standard statistical framework for distinguishing two distributions $\PP$ and $\QQ$. One chooses a real-valued test statistic $T=T(z)$, together with a threshold $\tau$, and outputs $\mathcal A_T(z):=\mathbf{1}\{T(z)\ge \tau\}$.
The guiding principle is that $T$ should have well-separated expectations under the planted and null models, while its fluctuations under both models remain smaller than this mean gap. In particular, whenever
\[
\bigl(\E_{\PP}[T]-\E_{\QQ}[T]\bigr)^2
\gg
\Var_{\PP}(T)+\Var_{\QQ}(T),
\]
the test $\mathcal A_T$ strongly distinguishes $\PP$ from $\QQ$. Thus, the main task in designing a detection algorithm is to construct a statistic $T$ whose mean is amplified under $\PP$ but remains small under $\QQ$, while simultaneously controlling its second moment. The hypergraph statistic introduced below is tailored precisely for this purpose.
\paragraph{Detection statistics: character sums over hypergraph families}
Our detection statistic is defined by aggregating character sums over a structured family $\mathcal H$ of $k$-uniform hypergraphs. We describe the statistic in the Bernoulli model, viewing the input as a function $z:\binom{[n]}{k}\to\{0,\pm1\}$. For each pattern $H\in\mathcal H$, we set
\[
T_H(z)
:=
\sum_{\phi:V(H)\hookrightarrow[n]}
\prod_{e\in E(H)} z(\phi(e)),
\]
where the sum ranges over injective embeddings of $H$ into the ambient vertex set $[n]$. We then define the global statistic
\[
F_{\mathcal H}(z):=\sum_{H\in\mathcal H} T_H(z).
\]

The key design choice is the pattern family $\mathcal H$. We take $\mathcal H$ to consist of $2$-regular $k$-uniform hypergraphs. This ensures that, under the planted model, every vertex appears an even number of times in the product over hyperedges, so the planted signs cancel identically:
\[
\prod_{e\in E(H)} x_{\phi(e)}
=
\prod_{u\in V(H)} x_{\phi(u)}^{\deg_H(u)}
=1,
\]
since $\deg_H(u)=2$ for every $u\in V(H)$. As a result, the expectation of the contribution from a copy of $H$ depends only on the sampling and noise parameters, and not on the unknown planted assignment. This cancellation phenomenon is the basic reason that $2$-regular hypergraphs are the correct objects for our statistic.

For the actual construction, we use the bounded-width hypergraph family introduced by Li~\cite{li2025smooth}, adapted here to the labeled setting. Working with labeled hypergraphs is technically convenient: every embedding on a prescribed vertex set is counted explicitly, and the resulting symmetry makes the moment computations more transparent. In particular, the null expectation vanishes, $\E_{\QQ}[F_{\mathcal H}(z)]=0$,
while under the planted model the statistic acquires a positive mean of order $\E_{\PP_x}[F_{\mathcal H}(z)]
=
|\mathcal H|(n)_v p^s\delta^s$,
when every $H\in\mathcal H$ has $v$ vertices and $s$ edges.

The main task is therefore to control the second moments. Expanding $F_{\mathcal H}(z)^2$ reduces the variance to counting pairs of embeddings with prescribed overlap structure, and the structure of the hypergraph family implies that these overlap counts are sparse enough. In the parameter regime of interest, this yields
\[
\frac{\Var_{\QQ}(F_{\mathcal H}(z))}{\E_{\PP_x}[F_{\mathcal H}(z)]^2}=o(1),
\qquad
\frac{\Var_{\PP_x}(F_{\mathcal H}(z))}{\E_{\PP_x}[F_{\mathcal H}(z)]^2}=o(1).
\]
By the standard mean--variance method, it follows that $F_{\mathcal H}(z)$ separates $\PP_x$ and $\QQ$ with vanishing error. Thus, the sample--time tradeoff in our algorithm is ultimately governed by the size and structure of the underlying bounded-width hypergraph family.
\paragraph{Recovery statistics and clean-up.}
The recovery statistic is a rooted variant of the detection statistic, based on a family $\mathcal J$ of bounded-width $k$-uniform hypergraphs in which all vertices have degree $2$ except for two distinguished leaves of degree $1$. For each ordered pair $a\neq b\in[n]$, we root these two leaves at $a$ and $b$, and define a character sum $F_{\mathcal J,a,b}(z)$ by aggregating over all rooted embeddings of patterns in $\mathcal J$. The key difference from detection is that the parity cancellation is now incomplete: all internal vertices still cancel, but the two leaves contribute
\[
\prod_{e\in E(J)} x_{\phi(e)}
=
x_ax_b,
\]
so under the planted model the statistic has mean proportional to $x_ax_b$, rather than a fixed positive bias. Thus the sign of $F_{\mathcal J,a,b}(z)$ predicts the pairwise product $x_ax_b$. As in the detection analysis, the main technical task is to control the second moment, and the bounded-width structure of $\mathcal J$ ensures that overlap patterns are sufficiently sparse, so that the statistic recovers $x_ax_b$ with probability $1-o(1)$ in the desired regime. Fixing an anchor vertex then yields a preliminary estimate $\hat x$ that is correct on a $1-o(1)$ fraction of coordinates, up to a global sign. Finally, we use a fresh independent batch of samples for a standard clean-up step: each coordinate is re-estimated by correlating the new constraints containing it against the preliminary assignment on the other $k-1$ variables. Since the preliminary estimate is already highly accurate, each such local vote is positively correlated with the true bit, and concentration upgrades weak recovery to exact recovery with high probability; when $k$ is odd, the same fresh samples also resolve the global sign ambiguity.
\paragraph{Algorithmic implementation: color coding and dynamic programming.}
Color coding, introduced by Alon, Yuster, and Zwick~\cite{alon1995color}, is a standard way to turn subgraph-counting problems into dynamic programs over color classes; it has since been used broadly in motif discovery and related counting tasks~\cite{alon2008biomolecular}. In average-case inference, closely related ideas have recently been combined with structured polynomial statistics, including Bayesian low-degree methods for stochastic block models~\cite{hopkins2017bayesian,mao2024testing,chen2026detecting}, smooth color-coded hypergraph statistics for tensor PCA and correlated spiked models~\cite{li2025smooth,li2025algorithmic}. We use the same general paradigm for noisy $k$-XOR. Since a direct computation of $F_{\mathcal H}(z)$ or $F_{\mathcal J,a,b}(z)$ would require enumerating all injective embeddings, we instead retain only \emph{colorful} embeddings under a random coloring $\tau:[n]\to[v]$, where $v$ is the number of vertices of the pattern. For any fixed embedding, the probability of being colorful is
$\rho=\frac{v!}{v^v}$,
so the colorful statistic is an unbiased estimator of the original one after rescaling by $1/\rho$; averaging over $t=\lceil 1/\rho\rceil$ independent colorings reduces the additional variance to a negligible level. The algorithmic advantage is that, once embeddings are required to be colorful, they can be computed by dynamic programming along the bounded-width decomposition of the pattern hypergraph. For detection, each $H\in\mathcal H$ decomposes into a cyclic chain of gadgets $U_1,\dots,U_\ell$; for recovery, each $J\in\mathcal J$ decomposes similarly into a path with two exposed leaves. In both cases, the dynamic program stores the contribution of a partial embedding as a function of its boundary vertices together with its used color set, and glues blocks together by summing over intermediate vertices and compatible color partitions. Since the width is controlled by the gadget size $D = rk$, the running time is dominated by the local enumeration inside one block, giving total time
\[
n^{D+3+o(1)}.
\]
Thus color coding serves two purposes simultaneously: it gives an unbiased approximation to the original character-sum statistic, and it converts the embedding problem into a bounded-width dynamic program, thereby yielding the claimed sample--time tradeoff.

\paragraph{On low-degree hardness.}
On the lower-bound side, we analyze the low-degree likelihood ratio in the Bernoulli sampling model. Since the coordinates are independent under the null law $\QQ$, the likelihood ratio $L=\mathrm{d}\PP/\mathrm{d}\QQ$ admits an orthogonal expansion in the natural product basis $\{\Phi_S\}$, indexed by sets $S$ of $k$-constraints. The Fourier coefficient of $L$ at $S$ is obtained by averaging the corresponding planted sign contribution over the hidden assignment $x$. This coefficient is nonzero exactly when every variable appears an even number of times across the constraints in $S$, or equivalently, when $S$ forms a $k$-uniform hypergraph in which every vertex has even degree. Thus the Fourier support of $L$ is precisely the family of vertex-even $k$-uniform hypergraphs. As a result, the degree-$D$ low-degree norm takes the form
\[
\|L_{\le D}\|_{L^2(\QQ)}^2
=
\sum_{\substack{S\text{ even}\\ |S|\le D}}
(p\delta^2)^{|S|},
\]
so the problem reduces to counting even $k$-uniform hypergraphs with at most $D$ edges. The main combinatorial input is an upper bound on the number of such hypergraphs: because every non-isolated vertex has degree at least $2$, a $t$-edge even hypergraph can involve only $O(kt)$ vertices, and a more careful counting argument yields the required estimate on their total number. Substituting this counting bound into the norm expression shows that the low-degree norm remains bounded, and in fact tends to $1$, whenever
\[
m\delta^2 \ll \frac{n^{k/2}}{D^{\,k/2-1}}.
\]
Under the low-degree heuristic, this implies that no algorithm running in time $n^{O(D)}$ should strongly distinguish $\PP$ from $\QQ$. This matches our upper bound up to the same $D^{k/2-1}$ factor, providing evidence that the algorithm is essentially optimal from the low-degree perspective.
\section{Preliminaries}
\paragraph{Basic notation.}
We write $\NN$ for the set of positive integers, and for $n\in\NN$ we write
$[n]:=\{1,2,\dots,n\}$. For a finite set $S$ and an integer $k\ge 0$, let
$\binom{S}{k}$ denote the family of all $k$-element subsets of $S$. For $k\ge 0$, we write $(n)_k:=n(n-1)\cdots(n-k+1)$,
with the convention $(n)_0:=1$.

We use the standard asymptotic notation $O(\cdot)$, $\Omega(\cdot)$,
$\Theta(\cdot)$, $o(\cdot)$, and $\omega(\cdot)$ with respect to the limit $n\to\infty$.
We also use $\widetilde O(\cdot)$, $\widetilde \Omega(\cdot)$, and
$\widetilde \Theta(\cdot)$ to suppress polylogarithmic factors in $n$.

\subsection{Hypergraph and isomorphism}
\paragraph{Hypergraphs.}
A \emph{$k$-uniform hypergraph} $H$ consists of a vertex set $V(H)$ and an edge set
$E(H)\subseteq \binom{V(H)}{k}$.
The \emph{degree} of a vertex $u\in V(H)$ is $\deg_H(u):=|\{e\in E(H): u\in e\}|$.
We say $H$ is \emph{$d$-regular} if $\deg_H(u)=d$ for all $u\in V(H)$.
We say $H$ is \emph{simple} if it has no repeated edges.


\paragraph{Isomorphism and labeled families.}
Let $H$ and $H'$ be $k$-uniform hypergraphs with $|V(H)|=|V(H')|=v$.
An \emph{isomorphism} from $H$ to $H'$ is a bijection $\pi:V(H)\to V(H')$ such that
\[
e\in E(H)\quad\Longleftrightarrow\quad \pi(e):=\{\pi(u):u\in e\}\in E(H').
\]
We write $H\cong H'$ if there exists an isomorphism between them.
When $V(H)=[v]$, we regard $H$ as a labeled hypergraph, and define
\[
\Aut(H)
:=
\Bigl\{\pi\in S_v:\ e\in E(H)\iff \pi(e):=\{\pi(u):u\in e\}\in E(H)\Bigr\}.
\]
be the automorphism group of $H$.
For two pattern hypergraphs $H,H'$ on $[v]$, let
\[
\Iso(H,H')
:=
\Bigl\{\pi\in S_v:\ e\in E(H)\iff \pi(e)\in E(H')\Bigr\}.
\]
be the set of graph isomorphisms from $H$ to $H'$.
In particular, $\Iso(H,H)=\Aut(H)$, and $\Iso(H,H')=\emptyset$ if $H\not\cong H'$.

\begin{definition}[Isomorphism-closed]\label{def:iso-closed}
Fix $v\in\mathbb N$ and view all patterns as labeled graphs on vertex set $[v]$.
A family $\mathcal H$ of labeled graphs on $[v]$ is called \emph{isomorphism-closed} if for every $H\in\mathcal H$ and every permutation
$\pi\in S_v$, the relabeled graph $\pi(H)$ also lies in $\mathcal H$, where
\[
V(\pi(H))=[v],
\qquad
E(\pi(H))=\{\pi(e):e\in E(H)\}.
\]
Equivalently, membership of $H$ in $\mathcal H$ depends only on the isomorphism type of $H$.
\end{definition}

Throughout the paper, the family $\mathcal H$ will be taken to be isomorphism-closed.

\begin{lemma}
Let $\mathcal H$ be an isomorphism-closed family of labeled $k$-uniform hypergraphs on $[v]$. Then for any fixed $H\in\mathcal H$,
\[
\sum_{H'\in\mathcal H} |\Iso(H,H')| = v!.
\]
\end{lemma}
\subsection{Statistical analysis framework}

\paragraph{Distributions.}
We write $\Ber(p)$ for the Bernoulli distribution on $\{0,1\}$ with
\[
\Pr[B=1]=p,
\qquad
\Pr[B=0]=1-p.
\]
We write $\Rad(1/2)$ for the Rademacher distribution on $\{\pm1\}$ with
\[
\Pr[Y=+1]=\Pr[Y=-1]=\frac12.
\]
More generally, $\Rad(\delta)$ denotes the $\{\pm1\}$-valued distribution with mean $\E[Y]=\delta$, namely
\[
\Pr[Y=+1]=\frac{1+\delta}{2},
\qquad
\Pr[Y=-1]=\frac{1-\delta}{2},
\qquad
\delta\in[-1,1].
\]
In particular, if $\delta=1-2\eta$, then $\Rad(\delta)$ corresponds to flipping a planted $\pm1$ label with probability $\eta$.

\begin{definition}
Let $\PP=\{\PP_n\}$ and $\QQ=\{\QQ_n\}$ be two sequences of distributions on $\Omega_n$, and let $\{\mathcal A_n\}$ be a sequence of tests, where each
\[
\mathcal A_n:\Omega_n\to\{0,1\}.
\]
We say that the sequence $\mathcal A_n$ \emph{weakly distinguishes} $\PP$ and $\QQ$ if
\[
\PP_n[\mathcal A_n=0]+\QQ_n[\mathcal A_n=1]\le 1-\Omega(1).
\]
We say that the sequence $\mathcal A_n$ \emph{strongly distinguishes} $\PP$ and $\QQ$ if
\[
\PP_n[\mathcal A_n=0]+\QQ_n[\mathcal A_n=1]=o(1).
\]
\end{definition}

\begin{lemma}\label{lem:mean-variance-separation}
Let $T_n:\Omega_n\to\mathbb R$ be any statistic with finite variance under both models, and assume that
\[
\E_{\PP_n}[T_n]>\E_{\QQ_n}[T_n].
\]
Define the midpoint threshold
\[
\tau_n:=\frac{\E_{\PP_n}[T_n]+\E_{\QQ_n}[T_n]}{2},
\]
and the test
\[
\mathcal A_n:=\mathbf 1\{T_n\ge \tau_n\}.
\]
Then
\[
\PP_n[\mathcal A_n=0]+\QQ_n[\mathcal A_n=1]
\le
\frac{
4\Var_{\PP_n}(T_n)+4\Var_{\QQ_n}(T_n)
}{
\bigl(\E_{\PP_n}[T_n]-\E_{\QQ_n}[T_n]\bigr)^2
}.
\]
In particular, if
\[
\bigl(\E_{\PP_n}[T_n]-\E_{\QQ_n}[T_n]\bigr)^2
\gg
\Var_{\PP_n}(T_n)+\Var_{\QQ_n}(T_n),
\]
then the sequence $\mathcal A_n$ strongly distinguishes $\PP$ from $\QQ$.
\end{lemma}

\begin{proof}
By Chebyshev's inequality,
\[
\PP_n[T_n<\tau_n]
=
\PP_n\!\left[
T_n-\E_{\PP_n}[T_n]
\le
-\frac{\E_{\PP_n}[T_n]-\E_{\QQ_n}[T_n]}{2}
\right]
\le
\frac{
4\Var_{\PP_n}(T_n)
}{
\bigl(\E_{\PP_n}[T_n]-\E_{\QQ_n}[T_n]\bigr)^2
}.
\]
Similarly,
\[
\QQ_n[T_n\ge\tau_n]
=
\QQ_n\!\left[
T_n-\E_{\QQ_n}[T_n]
\ge
\frac{\E_{\PP_n}[T_n]-\E_{\QQ_n}[T_n]}{2}
\right]
\le
\frac{
4\Var_{\QQ_n}(T_n)
}{
\bigl(\E_{\PP_n}[T_n]-\E_{\QQ_n}[T_n]\bigr)^2
}.
\]
Summing the two bounds proves the claim.
\end{proof}
\subsection{Equivalence of the fixed-\texorpdfstring{$m$}{m} and Bernoulli sampling models}\label{subsec:equi}

In this section, we compare three natural formulations of noisy $k$-XOR: the standard fixed-$m$ model with sampling with replacement, the fixed-$m$ model without replacement, and the Bernoulli-$p$ model. We will show that these models are equivalent for the purposes of our analysis.

Let
\[
\sM:=\binom{[n]}{k},
\qquad
M:=|\sM|=\binom{n}{k}.
\]
Fix an unknown planted vector $x\in\{\pm1\}^n$, and for each $\alpha\in\sM$, write
\[
x_\alpha:=\prod_{i\in\alpha}x_i\in\{\pm1\}.
\]
Let $\xi\sim\Rad(\delta)$; equivalently, the planted sign is flipped with probability $\eta=\frac{1-\delta}{2}$.
We now describe the three sampling models.

\paragraph{Model 1: fixed-$m$ sampling with replacement.}
We sample
\[
\alpha_1,\dots,\alpha_m \stackrel{\mathrm{i.i.d.}}{\sim} \Unif(\sM),
\]
and independently draw
\[
\xi_1,\dots,\xi_m \stackrel{\mathrm{i.i.d.}}{\sim} \Rad(\delta).
\]
Under the planted model $\PP_x$, the observation is the sequence
\[
\bigl((\alpha_t,z_t)\bigr)_{t=1}^m,
\qquad
z_t=x_{\alpha_t}\xi_t.
\]
Under the null model $\QQ$, the indices $\alpha_1,\dots,\alpha_m$ are sampled in the same way, while the labels $z(\alpha_t)$ are i.i.d. unbiased signs, independent of the sampled indices.

\paragraph{Model 2: fixed-$m$ sampling without replacement.}
We choose a uniformly random subset
\[
S_m\subseteq\sM,
\qquad
|S_m|=m.
\]
Under the planted model $\PP_x$, for each $\alpha\in S_m$ we observe
\[
z(\alpha)=x_\alpha \xi_\alpha,
\]
where the $(\xi_\alpha)_{\alpha\in S_m}$ are i.i.d. $\Rad(\delta)$. Under the null model $\QQ$, the same subset $S_m$ is sampled, but the labels $(z(\alpha))_{\alpha\in S_m}$ are independent unbiased signs.

\paragraph{Model 3: Bernoulli-$p$ sampling.}
For each $\alpha\in\sM$, independently sample
\[
b_\alpha\sim\Ber(p),
\qquad
\xi_\alpha\sim\Rad(\delta).
\]
Under the planted model $\PP_x$, we observe
\[
z(\alpha)=b_\alpha x_\alpha \xi_\alpha
\qquad
\text{for each }\alpha\in\sM.
\]
Equivalently, if
\[
S:=\{\alpha\in\sM:b_\alpha=1\},
\]
then the observed support is $S$, and each $\alpha\in S$ carries the label
\[
z(\alpha)=x_\alpha\xi_\alpha.
\]
Under the null model $\QQ$, the support $S$ is sampled in the same way, while conditional on $S$, the labels $(z(\alpha))_{\alpha\in S}$ are independent unbiased signs.

We will show that, in the nontrivial regime $m=o(M^{1/2})$,
and after matching the parameters via
\[
m=pM,
\qquad\text{equivalently}\qquad
p=\frac{m}{M},
\]
the above three models are asymptotically equivalent, up to an $o(m)$ discrepancy in the number of observed samples. Detailed proofs of these equivalences, together with the corresponding comparison results, are deferred to Appendix~\ref{appex:models}. In the rest of the paper, we work with the Bernoulli-$p$ model, with
\[
p=\frac{m}{\binom{n}{k}}.
\]
\section{Hypergraph Family and Statistics for Noisy $k$-XOR}

In this section, we introduce our hypergraph-based statistics for detection and recovery, and establish their theoretical guarantees.

\subsection{Hypergraph family for detection}
In this subsection, we describe the hypergraph family introduced in \cite{li2025smooth}, which was originally developed in the context of tensor PCA. In that work, the family is formulated in terms of unlabeled hypergraphs, and the relevant statistic is defined by counting labeled copies in $[n]$. Throughout this paper, by contrast, we work with labeled hypergraphs. This convention makes the notation cleaner and simplifies both the counting arguments and the analysis, while changing the resulting formulas only by the same combinatorial factors.

\begin{definition}\label{def:ugraph}
For any integer $r$, let $\mathcal U=\mathcal U(r,k)$ denote the family of labeled $k$-uniform hypergraphs $U$ such that:
\begin{enumerate}
\item $U$ has $rk+1$ vertices and $2r$ edges, and every vertex has degree either $1$ or $2$;
\item $U$ is connected;
\item deleting any single vertex leaves a connected hypergraph, that is, for every $v\in V(U)$, the hypergraph $U\setminus\{v\}$ remains connected.
\end{enumerate}
\end{definition}

Observe that
\[
\sum_{v\in V(U)} \deg_U(v)=k|E(U)|.
\]
Hence, condition~(1) implies that $U$ has exactly two vertices of degree $1$, while every other vertex has degree $2$. In other words, $U$ has exactly two leaves, and all remaining vertices are non-leaf vertices of degree $2$.

In \cite{li2025smooth}, Li proved the following lower bound on the size of $\mathcal U$.

\begin{lemma}\cite[Lemma 2.2]{li2025smooth}
For $k\ge 3$, there exists an absolute constant $C>0$ such that
\[
|\mathcal U|
\;\ge\;
C\cdot 
\frac{(2kr)!}{2^{rk}\,(k!)^{2r}\,(2r)!}.
\]
for sufficiently large $r>r(C)$.
\end{lemma}
Using these gadgets, we now construct the hypergraph family below.
\begin{definition}\label{def:family}
For any integers $r,\ell$, let $\mathcal H=\mathcal H(r,k,\ell)$ denote the family of labeled $k$-uniform hypergraphs $H$ for which there exist subhypergraphs
\[
U_1,\dots,U_\ell
\]
satisfying the following properties:
\begin{enumerate}
\item $U_i \in \mathcal U(r,k)$ for every $1\le i\le \ell$;
\item if the two leaves of $U_i$ are denoted by $L(U_i)=\{v_i,v_{i+1}\}$,
then
\[
V(U_i)\cap V(U_{i+1})=\{v_{i+1}\}
\qquad\text{for all }1\le i\le \ell,
\]
where we adopt the cyclic convention $v_{\ell+1}=v_1$ and $U_{\ell+1}=U_1$;
\item
\[
V(H)=\bigcup_{i=1}^{\ell} V(U_i),
\qquad
E(H)=\bigcup_{i=1}^{\ell} E(U_i).
\]
\end{enumerate}
\end{definition}

The above definition immediately implies that every $H\in\mathcal H$ has
\[
|V(H)|=rk\ell,\qquad |E(H)|=2r\ell,
\]
and every vertex has degree $2$, namely
\[
\deg_H(v)=2 \qquad\text{for all } v\in V(H).
\]
The next lemma gives a corresponding lower bound on the size of $\mathcal H$, showing that this family also contains sufficiently many labeled hypergraphs.

\begin{lemma}\cite[Lemma 2.4]{li2025smooth}
There exists an absolute constant $C>0$ such that for $k\ge 3$, 
\[
|\mathcal H|
\ge
\frac{1}{(2\ell)^2}\cdot
\left(C\cdot
\frac{(2kr)!}{2^{rk}\,(k!)^{r} \,(2r)!}
\right)^{\ell},
\]
for sufficiently large $r>r(C)$.
\end{lemma}

\subsection{Statistical analysis for detection}
As discussed in Section~\ref{subsec:equi}, the three sampling models introduced there are asymptotically equivalent in the regime relevant to our results. Accordingly, throughout the remainder of the paper we work with the Bernoulli sampling model.

Under the planted model $\PP_x$, for each $\alpha\in\binom{[n]}{k}$ we observe
\[
z(\alpha)=b_\alpha x_\alpha \xi_\alpha,
\]
where $b_\alpha\sim\Ber(p)$, and $\xi_\alpha\sim\Rad(\delta)$ independently over $\alpha$, and $x_\alpha:=\prod_{i\in\alpha}x_i$
for an unknown planted vector $x\in\{\pm1\}^n$.

Under the null model $\QQ$, for each $\alpha\in\binom{[n]}{k}$ we instead observe
\[
z(\alpha)=b_\alpha \zeta_\alpha,
\]
where $b_\alpha\sim\Ber(p)$ and
$\zeta_\alpha\sim\Rad(0)$
independently over $\alpha$.
\paragraph{Pattern family and statistic.}
Let $\mathcal H$ be a finite, isomorphism-closed family of labeled $k$-uniform,
$2$-regular hypergraphs. For simplicity, assume that every $H\in\mathcal H$ is
\emph{simple} (i.e., has no repeated edges) and shares the same parameters
$|V(H)|=v$,
$|E(H)|=s$.
For each $H\in\mathcal H$, define the embedding statistic
\[
T_H(z)
\;:=\;
\sum_{\varphi:V(H)\hookrightarrow[n]}
\ \prod_{e\in E(H)} z(\varphi(e)),
\qquad
F_{\mathcal H}(z)
\;:=\;
\sum_{H\in\mathcal H} T_H(z),
\]
where the sum ranges over injections $\varphi:V(H)\hookrightarrow[n]$ and we write
$\varphi(e):=\{\varphi(u):u\in e\}\in\binom{[n]}{k}$ for the induced image of an edge.

Next we analyze the first- and second-moment behavior under these two distributions. Here $x$ is treated as fixed throughout, and the randomness is over $(b_\beta,\xi_\beta)_{\beta\in\binom{[n]}{k}}$ and $(b_\beta,\zeta_\beta)_{\beta\in\binom{[n]}{k}}$.
\paragraph{Expectation under $\QQ$.}
Fix $H\in\mathcal H$ and an injection $\varphi:V(H)\hookrightarrow[n]$.
Since $H$ is simple and $\varphi$ is injective, the $k$-sets
$\{\varphi(e):e\in E(H)\}$ are all distinct. Under $\QQ$, the random variables
$\{z(\alpha)\}_{\alpha\in\mathcal A}$ are independent and satisfy
$\E_{\QQ}[z(\alpha)]=0$. Hence
\[
\E_{\QQ}\!\left[\prod_{e\in E(H)} z(\varphi(e))\right]
=
\prod_{e\in E(H)} \E_{\QQ}[z(\varphi(e))]
=
0.
\]
Summing over embeddings gives
\[
\E_{\QQ}[T_H(z)]=0,
\qquad
\E_{\QQ}[F_{\mathcal H}(z)]=0.
\]

\paragraph{Expectation under $\PP_x$}
Fix $H\in\mathcal H$ and an injection $\varphi:V(H)\hookrightarrow[n]$.
Under $\PP_x$ we have $\E_{\PP_x}[z(\alpha)]=p\delta x_\alpha$ for each $\alpha$, and the family $\{z(\alpha)\}$
is independent across distinct $\alpha$. Therefore
\[
\begin{aligned}
\E_{\PP_x}\!\left[\prod_{e\in E(H)} z(\varphi(e))\right]
&=
\prod_{e\in E(H)} \E_{\PP_x}[z(\varphi(e))]
=
\prod_{e\in E(H)} (p\delta x_{\varphi(e)})\\
&=
(p\delta)^s \cdot \prod_{e\in E(H)} x_{\varphi(e)}.
\end{aligned}
\]
Moreover,
\[
\prod_{e\in E(H)} x_{\varphi(e)}
=
\prod_{e\in E(H)} \prod_{u\in e} x_{\varphi(u)}
=
\prod_{u\in V(H)} x_{\varphi(u)}^{\deg_H(u)}.
\]
Since each $H\in\mathcal H$ is $2$-regular, $\deg_H(u)=2$ for all $u$ and hence
$x_{\varphi(u)}^{\deg_H(u)}=1$. It follows that
\[
\E_{\PP_x}\!\left[\prod_{e\in E(H)} z(\varphi(e))\right]
=
(p\delta)^s.
\]
Summing over the $(n)_v$ injections $\varphi:V(H)\hookrightarrow[n]$ yields
\[
\E_{\PP_x}[T_H(z)]=(n)_v\,(p\delta)^s,
\qquad
\E_{\PP_x}[F_{\mathcal H}(z)]=|\mathcal H|\,(n)_v\,(p\delta)^s.
\]

\paragraph{Variance under $\QQ$}
Since $\E_{\QQ}[F_{\mathcal H}(z)]=0$, we have
\[
\Var_{\QQ}(F_{\mathcal H}(z))=\E_{\QQ}\!\big[F_{\mathcal H}(z)^2\big].
\]
Expanding the square gives
\begin{equation}\label{eq:F-second-expand}
\E_{\QQ}[F_{\mathcal H}(z)^2]
=\sum_{H,H'\in\mathcal H}\ \sum_{\varphi:V(H)\hookrightarrow[n]}\ \sum_{\psi:V(H')\hookrightarrow[n]}
\E_{\QQ}\!\left[\ \prod_{e\in E(H)} z(\varphi(e))\ \prod_{e'\in E(H')} z(\psi(e'))\right].
\end{equation}
Fix $H,H'\in\mathcal H$ and injections $\varphi,\psi$.
Define the induced multiset of $k$-sets
\[
\mathcal M(\varphi,\psi)
:=\{\varphi(e):e\in E(H)\}\ \cup\ \{\psi(e'):e'\in E(H')\},
\]
and let $m_\alpha(\varphi,\psi)$ denote the multiplicity of $\alpha$ in
$\mathcal M(\varphi,\psi)$. By independence across distinct $\alpha$,
\[
\E_{\QQ}\!\left[\ \prod_{e\in E(H)} z(\varphi(e))\ \prod_{e'\in E(H')} z(\psi(e'))\right]
=
\prod_{\alpha\in\mathcal A}\E_{\QQ}\!\big[z(\alpha)^{m_\alpha(\varphi,\psi)}\big].
\]
If there exists $\alpha$ with $m_\alpha(\varphi,\psi)$ odd, then
$\E_{\QQ}[z(\alpha)^{m_\alpha(\varphi,\psi)}]=0$ and the whole product vanishes.
Thus the expectation can be nonzero only if every induced $k$-set appears with even
multiplicity and each term is nonzero. Then each multiplicity
$m_\alpha(\varphi,\psi)\in\{0,1,2\}$, so the even-multiplicity condition is equivalent to
\begin{equation}\label{eq:edge-set-match}
\{\varphi(e):e\in E(H)\}=\{\psi(e'):e'\in E(H')\}
\qquad\text{as (unordered) sets.}
\end{equation}
When \eqref{eq:edge-set-match} holds, each of the $s$ distinct induced edges appears
exactly twice, and hence
\begin{equation}\label{eq:match-contribution}
\E_{\QQ}\!\left[\ \prod_{e\in E(H)} z(\varphi(e))\ \prod_{e'\in E(H')} z(\psi(e'))\right]
=
\prod_{e\in E(H)} \E_{\QQ}[z(\varphi(e))^2]
=
p^s.
\end{equation}
Because $\mathcal H$ is isomorphism-closed on $[v]$, for any fixed $H\in\mathcal H$
and injection $\varphi:V(H)\hookrightarrow[n]$, the number of pairs $(H',\psi)$ with
$H'\in\mathcal H$ and $\psi:V(H')\hookrightarrow[n]$ satisfying
\eqref{eq:edge-set-match} is exactly the number of vertex relabelings of $H$, namely $v!$.
Equivalently, defining
\[
\Gamma_{\mathcal H}(H,\varphi)
:=\#\Big\{(H',\psi):\ H'\in\mathcal H,\ \psi:V(H')\hookrightarrow[n],\
\{\psi(e')\}_{e'\in E(H')}=\{\varphi(e)\}_{e\in E(H)}\Big\},
\]
we have
\begin{equation}\label{eq:Gamma-vfact}
\Gamma_{\mathcal H}(H,\varphi)=v!.
\end{equation}
Plugging \eqref{eq:match-contribution} and \eqref{eq:Gamma-vfact} into
\eqref{eq:F-second-expand} yields
\[
\E_{\QQ}[F_{\mathcal H}(z)^2]
=\sum_{H\in\mathcal H}\ \sum_{\varphi:V(H)\hookrightarrow[n]} \Gamma_{\mathcal H}(H,\varphi)\cdot p^s
=\sum_{H\in\mathcal H} (n)_v\cdot v!\cdot p^s
=|\mathcal H|\,(n)_v\,v!\,p^s.
\]
Therefore,
\[
\Var_{\QQ}(F_{\mathcal H}(z)) = |\mathcal H|\,(n)_v\,v!\,p^s.
\]

\paragraph{Variance under $\PP_x$}
Fix any secret $x\in\{\pm 1\}^n$. We have
\[
\Var_{\PP_x}\!\bigl(F_{\mathcal H}(z)\bigr)
=
\E_{\PP_x}\!\bigl[F_{\mathcal H}(z)^2\bigr]
-
\E_{\PP_x}\!\bigl[F_{\mathcal H}(z)\bigr]^2.
\]
Expanding the second moment gives
\[
\E_{\PP_x}\!\bigl[F_{\mathcal H}(z)^2\bigr]
=
\sum_{H,H'\in\mathcal H}
\ \sum_{\varphi:V(H)\hookrightarrow[n]}
\ \sum_{\psi:V(H')\hookrightarrow[n]}
\E_{\PP_x}\!\left[
\prod_{e\in E(H)} z\!\bigl(\varphi(e)\bigr)
\prod_{e'\in E(H')} z\!\bigl(\psi(e')\bigr)
\right].
\]
Fix $H,H'\in\mathcal H$ and injective maps
$\varphi:V(H)\hookrightarrow[n]$, $\psi:V(H')\hookrightarrow[n]$.
Let
\[
E(\varphi(H)):=\{\varphi(e):e\in E(H)\},
\qquad
E(\psi(H')):=\{\psi(e'):e'\in E(H')\}
\]
denote the corresponding image edge sets, viewed as collections of $k$-subsets of $[n]$. Define
\[
j:=\bigl|E(\varphi(H))\cap E(\psi(H'))\bigr|.
\]
Since $|E(H)|=|E(H')|=s$, it follows that
\[
\bigl|E(\varphi(H))\cup E(\psi(H'))\bigr|=2s-j,
\qquad
\bigl|E(\varphi(H))\triangle E(\psi(H'))\bigr|=2(s-j).
\]
Since $H$ and $H'$ are simple and $\varphi,\psi$ are injective, every induced $k$-edge in
$E(\varphi(H))\cup E(\psi(H'))$
appears with multiplicity at most $2$. Using $b_\alpha^2=b_\alpha$,
$x_\alpha^2=\xi_\alpha^2=1$,
together with independence across distinct $\alpha$, we obtain
\[
\E_{\PP_x}\!\left[
\prod_{e\in E(H)} z\!\bigl(\varphi(e)\bigr)
\prod_{e'\in E(H')} z\!\bigl(\psi(e')\bigr)
\right]
=
p^{\,|E(\varphi(H))\cup E(\psi(H'))|}
\delta^{\,|E(\varphi(H))\triangle E(\psi(H'))|}
\prod_{\alpha\in E(\varphi(H))\triangle E(\psi(H'))} x_\alpha.
\]
In particular, since the symmetric difference $E(\varphi(H))\triangle E(\psi(H'))$
is vertex-even, then
\[
\E_{\PP_x}\!\left[
\prod_{e\in E(H)} z\!\bigl(\varphi(e)\bigr)
\prod_{e'\in E(H')} z\!\bigl(\psi(e')\bigr)
\right]
=
p^{\,2s-j}\delta^{\,2s-2j},
\]
where
\[
j=\bigl|E(\varphi(H))\cap E(\psi(H'))\bigr|.
\]
Accordingly, we may write
\[
\Var_{\PP_x}\!\bigl(F_{\mathcal H}(z)\bigr)
=
\sum_{H,H'\in\mathcal H}\ \sum_{\varphi}\ \sum_{\psi}
\left(
p^{\,2s-j(\varphi,\psi)}\delta^{\,2s-2j(\varphi,\psi)}
-(p\delta)^{2s}
\right),
\]
where
\[
j(\varphi,\psi):=\bigl|E(\varphi(H))\cap E(\psi(H'))\bigr|.
\]
For a pair of injections $\varphi:V(H)\hookrightarrow[n]$ and $\psi:V(H')\hookrightarrow[n]$,
let
\[
i:=|\varphi(V(H))\cap \psi(V(H'))|\in\{0,1,\dots,v\}.
\]
The number of ordered pairs of vertex images with intersection size $i$ equals
\[
N(i,j)
:=
\#\Bigl\{
(H,H',\varphi,\psi):
|\varphi(V(H))\cap \psi(V(H'))|=i,\ 
|E(\varphi(H))\cap E(\psi(H'))|=j
\Bigr\}.
\]
Thus
\begin{equation}\label{eq:VarP-i-j}
\begin{aligned}
\Var_{\PP_x}\!\big(F_{\mathcal H}(z)\big)
&=
\sum_{i=0}^{v} 
\sum_{j=0}^{s}N(i,j)\Big(p^{2s-j}\delta^{2s-2j}-(p\delta)^{2s}\Big)\\
&\leq
\sum_{i=1}^{v} 
\sum_{j=1}^{s} N(i,j)\cdot p^{2s-j}\delta^{2s-2j}.
\end{aligned}
\end{equation}

Since every $H\in\mathcal H$ is connected, any proper subset of $i$ vertices must be incident to at least one hyperedge leaving the subset. This yields the following structural restriction.

\begin{lemma}\label{lem:Nij-zero}
For any pair $(i,j)$ with
\[
2 i\leq jk+1,
\]
we have
\[
N(i,j)=0,
\]
except possibly in the diagonal case $(i,j)=(v,s)$. Moreover,
\[
N(v,s)=|\mathcal H|\,(n)_v\,v!.
\]
\end{lemma}

We also record the following crude counting bound.

\begin{lemma}\label{lem:Nij-sum}
For every $1\le i<v$,
\[
\sum_{j=1}^{\lfloor(2i-2)/k\rfloor} N(i,j)
\le
|\mathcal H|^2\,(n)_v\,(n)_{v-i}\binom{v}{i}.
\]
\end{lemma}

This lemma is simply about the number of pairs of embeddings whose vertex intersection has size exactly $i$: one first chooses $H,H'\in\mathcal H$, then an embedding of $H$, then the $i$ intersecting vertices, and finally the remaining $v-i$ vertices of the embedding of $H'$. Since this is a straightforward combinatorial count, we omit the proof.

\begin{lemma}[Mean--variance criterion for $F_{\mathcal H}$]\label{lem:Meanvariance}
Fix any constant $\epsilon>0$. Let $\mathcal H$ be the family of hypergraphs from Definition~\ref{def:family}, where $v=rk\ell$, $s=2r\ell$. Suppose $v^v = O(n^{0.3})$. Then consider the threshold test
\begin{equation}\label{eq:Ag}
\mathcal A(z)
\;=\;
\begin{cases}
1, & F_{\mathcal H}(z)\ge \dfrac12\,\E_{\PP_x}[F_{\mathcal H}],\\
0, & F_{\mathcal H}(z)< \dfrac12\,\E_{\PP_x}[F_{\mathcal H}].
\end{cases}
\end{equation}
Then $\mathcal A$ distinguishes $\PP_x$ from $\QQ$ with success probability at least
\[
1-
O\!\left(
\left[
\frac{4\,\ell^{k/2} k^{k/2}}{k!\,2^{k/2}}
\cdot
\frac{n^{k/2}}{m\,\delta^2\,(kr)^{k/2-1}}
\right]^{2r\ell}
+
\frac{1}{n^{0.1}}
\left(
\frac{n^{k/2}}{k!\,m\,\delta^2\,(rk)^{k/2-1}}
\right)^{2r\ell}
\right).
\]
\end{lemma}

The proof is deferred to Appendix~\ref{appex:mean-var}.

\begin{theorem}\label{thm:mean-var}
Assume $r$ is larger than a sufficiently large absolute constant, set $D:=rk$ and let $\ell\ge 2$ satisfy $D\ell<\frac{0.3\log n}{\log\log n}$.
If
\[
m
\ge
\frac{4\,\ell^{k}k^{k/2}}{k!\,2^{k/2}}
\cdot
\frac{n^{k/2}}{\delta^2\,D^{\,k/2-1}},
\]
then the test $\mathcal A$ in \eqref{eq:Ag} distinguishes $\PP_x$ from $\QQ$ with success probability at least
\[
1-O\!\left(\ell^{-kD}\right).
\]
Therefore, $\mathcal A$ weakly distinguishes $\PP_x$ from $\QQ$. Moreover, if at least one of the parameters $\ell$ or $D$ satisfies $\omega(1)$ as $n\to\infty$, then $\mathcal A$ in fact strongly distinguishes $\PP_x$ from $\QQ$.
\end{theorem}

\subsection{Hypergraph family for recovery}
Starting from the gadget family $\mathcal U(r,k)$ in Definition~\ref{def:ugraph}, we now define the hypergraph family used for recovery. This is the path analogue of the detection family and also appears in \cite{li2025smooth}.

\begin{definition}\label{def:family-recovery}
For any integers $r,\ell$, let $\mathcal J=\mathcal J(r,k,\ell)$ denote the family of labeled $k$-uniform hypergraphs $J$ such that there exist subhypergraphs
\[
U_1,\dots,U_\ell
\]
with the following properties:
\begin{enumerate}
\item $U_i\in \mathcal U(r,k)$ given in Definition \ref{def:ugraph} for all $1\le i\le \ell$;

\item if $L(U_i)=\{v_i,v_{i+1}\}$, then
\[
V(U_i)\cap V(U_{i+1})=\{v_{i+1}\}
\qquad\text{for all }1\le i\le \ell-1;
\]

\item
\[
V(J)=\bigcup_{i=1}^{\ell}V(U_i),
\qquad
E(J)=\bigcup_{i=1}^{\ell}E(U_i);
\]
\end{enumerate}
\end{definition}

By construction, every $J\in\mathcal J$ has
\[
|V(J)|=rk\ell+1,
\qquad
|E(J)|=2r\ell,
\qquad
L(J)=\{v_1,v_{\ell+1}\}.
\]
Moreover, every internal vertex has degree $2$, while the two endpoint vertices $v_1$ and $v_{\ell+1}$ have degree $1$.

The family $\mathcal J$ is the path analogue of the detection family $\mathcal H$: instead of gluing the gadgets cyclically, we glue them into a chain with two exposed leaves, which is the natural structure for the rooted recovery statistic. 

In \cite{li2025smooth}, the author proved a rough lower bound for the unlabeled version of
$\mathcal J$. Since we work with the labeled family, we record here a corresponding lower bound, analogous to the one for
$\mathcal H$.

\begin{lemma}\cite[Lemma 2.9]{li2025smooth}\label{lem:counting-J}
There exists an absolute constant $C>0$ such that, for every $k\ge 3$,
\[
|\mathcal J|
\;\ge\;
\frac{1}{(2\ell)^2}\cdot
\left(C\cdot
\frac{(2kr)!}{2^{rk}\,(k!)^{r}\,(2r)!}
\right)^{\ell},
\]
for all sufficiently large $r\ge r_0(C)$.
\end{lemma}

\begin{proof}
For every $H\in\mathcal H$, disconnecting the identified vertex $v_1=v_{\ell+1}$ produces a
hypergraph $J\in\mathcal J$. Thus each member of $\mathcal H$ naturally gives rise to a member
of $\mathcal J$, and the claimed lower bound follows from the corresponding counting bound for
$\mathcal H$.
\end{proof}
\subsection{Statistical analysis for recovery}

Let $\mathcal J=\mathcal J(r,k,\ell)$ be the recovery family from
Definition~\ref{def:family-recovery}. Every $J\in\mathcal J$ has
\[
|V(J)|=v:=rk\ell+1,
\qquad
|E(J)|=s:=2r\ell,
\]
and exactly two leaves, denoted by $a=a(J)$ and $b=b(J)$.

For each ordered pair $(a,b)\in[n]^2$ with $a\neq b$, define
\[
T_{J,a,b}(z)
:=
\sum_{\substack{\phi:V(J)\hookrightarrow[n]\\ \phi(v_1)=a,\ \phi(v_{\ell+1})=b}}
\ \prod_{e\in E(J)} z(\phi(e)).
\]
Then define the aggregated rooted statistic
\[
F_{\mathcal J,a,b}(z)
:=
\sum_{J\in\mathcal J} T_{J,a,b}(z).
\]

We now compute the statistical properties of this statistic under the planted model.
\paragraph{Expectation}
Fix $J\in\mathcal J$ and an injection $\phi:V(J)\hookrightarrow[n]$ satisfying
$\phi(v_1)=a$ and $\phi(v_{\ell+1})=b$. Since the image hyperedges
$\{\phi(e):e\in E(J)\}$ are distinct, independence gives
\[
\E\!\left[\prod_{e\in E(J)} z(\phi(e))\right]
=
\prod_{e\in E(J)} \E[z(\phi(e))]
=
\prod_{e\in E(J)} p\delta\,x_{\phi(e)}.
\]
Hence
\[
\E\!\left[\prod_{e\in E(J)} z(\phi(e))\right]
=
(p\delta)^s \prod_{e\in E(J)} x_{\phi(e)}.
\]
Now
\[
\prod_{e\in E(J)} x_{\phi(e)}
=
\prod_{u\in V(J)} x_{\phi(u)}^{\deg_J(u)}.
\]
Since every internal vertex of $J$ has degree $2$, while the two leaves
$v_1$ and $v_{\ell+1}$ have degree $1$, all internal signs cancel, and thus
\[
\prod_{e\in E(J)} x_{\phi(e)}
=
x_{\phi(v_1)}x_{\phi(v_{\ell+1})}
=
x_ax_b.
\]
Therefore
\[
\E\!\left[\prod_{e\in E(J)} z(\phi(e))\right]
=
(p\delta)^s\,x_ax_b.
\]

The number of injections $\phi:V(J)\hookrightarrow[n]$ such that
$\phi(v_1)=a$ and $\phi(v_{\ell+1})=b$ is $(n-2)_{v-2}$.
Thus
\[
\E[T_{J,a,b}(z)]
=
(n-2)_{v-2}(p\delta)^s\,x_ax_b.
\]
Summing over $J\in\mathcal J$, we obtain
\begin{equation}\label{eq:E-FJab}
\E[F_{\mathcal J,a,b}(z)]
=
|\mathcal J|(n-2)_{v-2}(p\delta)^s\,x_ax_b.
\end{equation}

\paragraph{Variance under the planted model}

We write
\[
\Var(F_{\mathcal J,a,b}(z))
=
\E[F_{\mathcal J,a,b}(z)^2]
-
\E[F_{\mathcal J,a,b}(z)]^2.
\]
Expanding the square,
\[
\E[F_{\mathcal J,a,b}(z)^2]
=
\sum_{J,J'\in\mathcal J}
\ \sum_{\substack{\phi:V(J)\hookrightarrow[n]\\ \phi(v_1)=a,\ \phi(v_{\ell+1})=b}}
\ \sum_{\substack{\psi:V(J')\hookrightarrow[n]\\ \psi(v_1)=a,\ \psi(v_{\ell+1})=b}}
\E\!\left[
\prod_{e\in E(J)} z(\phi(e))
\prod_{e'\in E(J')} z(\psi(e'))
\right].
\]
For such a pair $(J,\phi)$, $(J',\psi)$, let
\[
i:=|\phi(V(J))\cap \psi(V(J'))|,
\qquad
j:=|E(\phi(J))\cap E(\psi(J'))|.
\]
Since both rooted embeddings send the two leaves to $a,b$, we always have $i\ge 2$.

As in the detection calculation, every hyperedge in the union contributes one factor $p$, and
every hyperedge in the symmetric difference contributes one factor $\delta$. Hence
\[
\E\!\left[
\prod_{e\in E(J)} z(\phi(e))
\prod_{e'\in E(J')} z(\psi(e'))
\right]
=
p^{\,2s-j}\delta^{\,2s-2j}
\prod_{\alpha\in E(\phi(J))\triangle E(\psi(J'))} x_\alpha.
\]
Since each rooted copy contributes the endpoint parity $x_ax_b$, we have
\[
\prod_{\alpha\in E(\phi(J))\triangle E(\psi(J'))} x_\alpha
=
(x_ax_b)^2
=
1.
\]
Therefore
\[
\E\!\left[
\prod_{e\in E(J)} z(\phi(e))
\prod_{e'\in E(J')} z(\psi(e'))
\right]
=
p^{\,2s-j}\delta^{\,2s-2j}.
\]

For $2\le i\le v$ and $0\le j\le s$, let $N_{a,b}(i,j)$ denote the number of ordered tuples
\[
(J,J',\phi,\psi)
\]
such that
\[
\phi(v_1)=\psi(v_1)=a,\qquad \phi(v_{\ell+1})=\psi(v_{\ell+1})=b,
\]
and
\[
|\phi(V(J))\cap\psi(V(J'))|=i,
\qquad
|E(\phi(J))\cap E(\psi(J'))|=j.
\]
Then
\[
\E[F_{\mathcal J,a,b}(z)^2]
=
\sum_{i=2}^{v}\sum_{j=0}^{s}
N_{a,b}(i,j)\,p^{\,2s-j}\delta^{\,2s-2j}.
\]
Since
\[
\E[F_{\mathcal J,a,b}(z)]
=
|\mathcal J|(n-2)_{v-2}(p\delta)^s\,x_ax_b,
\]
we obtain
\[
\Var(F_{\mathcal J,a,b}(z))
=
\sum_{i=2}^{v}\sum_{j=0}^{s}
N_{a,b}(i,j)\Bigl(p^{2s-j}\delta^{2s-2j}-(p\delta)^{2s}\Bigr).
\]
In particular,
\[
\Var(F_{\mathcal J,a,b}(z))
\le
\sum_{i=2}^{v}\sum_{j=1}^{s}
N_{a,b}(i,j)\, p^{\,2s-j}\delta^{\,2s-2j}.
\]

\begin{lemma}\label{lem:Nij-zero-recovery}
For every pair $(i,j)$ with
$$
2i\leq jk+5,
$$
we have
\[
N_{a,b}(i,j)=0,
\]
except for the diagonal case $(i,j)=(v,s)$. Moreover
\[
N_{a,b}(v,s)=|\mathcal J|\,(n-2)_{v-2}\,(v-2)!.
\]
\end{lemma}
The proof is deferred to Appendix \ref{appex:proof-Nij-recovery}
\begin{lemma}\label{lem:rooted-overlap-fixed-i}
For every $2\le i < v$,
\[
\sum_{j=1}^{\lfloor (2i-6)/k\rfloor} N_{a,b}(i,j)
\;\le\;
|\mathcal J|^2\,(n-2)_{v-2}\binom{v-2}{i-2}(n-v)_{v-i}.
\]
\end{lemma}
\begin{proof}
This is a crude counting bound: first choose $J,J'\in\mathcal J$, then choose the rooted embedding $\phi$ of $J$, which contributes $ |\mathcal J|^2 (n-2)_{v-2} $ choices since the two endpoints are fixed to $a,b$.  
Next, to define $\psi$, choose the other $i-2$ shared vertices among the $v-2$ internal vertices of $\phi(J)$, and then map the remaining $v-i$ vertices outside $\phi(V(J))$, giving at most $ \binom{v-2}{i-2}(n-v)_{v-i} $ possibilities.  
Multiplying these bounds gives the claim, and the restriction $ j\le \lfloor (2i-6)/k\rfloor $ follows from Lemma~\ref{lem:Nij-zero-recovery}.
\end{proof}

Thus the recovery statistic satisfies the same type of mean--variance guarantee as in the detection setting.

\begin{lemma}\label{lem:Meanvariance-recovery}
Fix any constant $\varepsilon>0$. Let $\mathcal J$ be the family of hypergraphs from
Definition~\ref{def:family-recovery}, where
$v=rk\ell+1$, $s=2r\ell$. Suppose $v^v = O(n^{0.3})$, then for each ordered pair $a\neq b$, define
\[
\mathcal A_{a,b}(z)
\;:=\;
\begin{cases}
1, & F_{\mathcal J,a,b}(z)\ge 0,\\[4pt]
-1, & F_{\mathcal J,a,b}(z)<0.
\end{cases}
\]
Then $\mathcal A_{a,b}$ outputs $x_ax_b$ with probability at least
\begin{equation}\label{eq:probrecovery}
1-
O\!\left(
\left[
\frac{4\,\ell^{k/2} k^{k/2}}{k!\,2^{k/2}}
\cdot
\frac{n^{k/2}}{m\,\delta^2\,(kr)^{k/2-1}}
\right]^{2r\ell}
+
\frac{1}{n^{0.1}}
\left(
\frac{n^{k/2}}{k!\,m\,\delta^2\,(rk)^{k/2-1}}
\right)^{2r\ell}
\right).
\end{equation}

Moreover, fix an anchor vertex $1\in[n]$, and define
\[
\widehat x_1:=1,
\qquad
\widehat x_b:=\sgn\!\bigl(F_{\mathcal J,1,b}(z)\bigr)
\quad\text{for all } b\in[n]\setminus\{1\}.
\]
Then, with probability at least
\begin{equation}\label{eq:probrecovery2}
1-
O\!\left(
k\left[
\frac{4\,\ell^{k/2} k^{k/2}}{k!\,2^{k/2}}
\cdot
\frac{n^{k/2}}{m\,\delta^2\,(kr)^{k/2-1}}
\right]^{2r\ell}
+
\frac{k}{n^{0.1}}
\left(
\frac{n^{k/2}}{k!\,m\,\delta^2\,(rk)^{k/2-1}}
\right)^{2r\ell}
\right),
\end{equation}
the estimator $\widehat x$ correctly recovers at least
\[
\left(1-\frac{1}{4k}\right)n
\]
coordinates, up to the global sign.
\end{lemma}

The proof is deferred to Appendix~\ref{apx:recover-proof}

\subsection{Rounding and clean-up}
For the final clean-up step, we draw an additional independent set of fresh samples of size $m':=n^{1.1}$.
This guarantees that the clean-up stage is independent of the preliminary estimate $\hat x$, and since $m'=n^{1.1}=o(m)$ in our parameter regime, it increases the total sample complexity only by a lower-order term.

We first resolve the global sign ambiguity. When $k$ is even, the planted $k$-XOR model is invariant under the transformation $x\mapsto -x$, so recovering either $x$ or $-x$ is sufficient. When $k$ is odd, we use the fresh samples to determine the correct sign: if the fraction of satisfied equations under $\hat x$ is less than $1/2$, we replace $\hat x$ by $-\hat x$.

Now suppose that, after this sign correction when $k$ is odd, we are given an estimate
$\hat x\in\{\pm1\}^n$ such that
\[
\frac1n\bigl|\{i\in[n]:\hat x_i\neq x_i\}\bigr|\le \eta
\qquad\text{for}\qquad
\eta = \frac{1}{4k}.
\]
We generate the clean-up sample set by drawing $m'=n^{1.1}$ fresh constraints independently and uniformly from $\binom{[n]}{k}$. Equivalently, in the Bernoulli model each $k$-tuple $e\in\binom{[n]}{k}$ is included independently with probability
\[
p':=\frac{m'}{\binom{n}{k}},
\]
and, conditioned on being included, its observed label is
\[
y_e=\xi_e\prod_{j\in e}x_j,
\]
where
\[
\Pr[\xi_e=1]=\frac{1+\delta}{2},
\qquad
\Pr[\xi_e=-1]=\frac{1-\delta}{2}.
\]
For each $i\in[n]$, define
\[
T_i:=\sum_{e\ni i} y_e \prod_{j\in e\setminus\{i\}}\hat x_j,
\qquad
\widetilde x_i:=\sgn(T_i),
\]
with ties broken arbitrarily. The full recovery algorithm is given in Algorithm~\ref{alg:full-recovery}.

\begin{algorithm}[htbp]
\DontPrintSemicolon
\LinesNumbered
\caption{Recovery algorithm with sample splitting and clean-up}\label{alg:full-recovery}
\KwIn{
A collection of observed noisy $k$-XOR samples
$\mathcal E=\{(e,y_e)\}$,
where each $e\in\binom{[n]}{k}$ is observed together with its label $y_e\in\{\pm1\}$.
}
\KwOut{
A recovered assignment $\widetilde x\in\{\pm1\}^n$.
}

Randomly split $\mathcal E$ into two disjoint parts
$\mathcal E_{\mathrm{main}}$ and
$\mathcal E_{\mathrm{clean}}$,
by assigning each observed sample independently to $\mathcal E_{\mathrm{clean}}$ with probability $\rho$, and otherwise to $\mathcal E_{\mathrm{main}}$\;

Using only $\mathcal E_{\mathrm{main}}$, compute a preliminary estimate $\hat x\in\{\pm1\}^n$
via the character-sum procedure from Algorithm~\ref{alg:rec-final} and rounding procedure in Lemma \ref{lem:Meanvariance-recovery}\;

\If{$k$ is odd}{
    Compute the fraction of equations in $\mathcal E_{\mathrm{clean}}$ satisfied by $\hat x$\;
    
    \If{this fraction is less than $1/2$}{
        Replace $\hat x\gets -\hat x$\;
    }
}

\For{$i=1,2,\dots,n$}{
    Set
    $T_i\gets\sum_{\substack{(e,y_e)\in\mathcal E_{\mathrm{clean}}\\ i\in e}}
    y_e\prod_{j\in e\setminus\{i\}}\hat x_j
    $

    Set
    $\widetilde x_i\gets \sgn(T_i)$
    with ties broken arbitrarily\;
}

\Return $\widetilde x$\;
\end{algorithm}
\begin{lemma}\label{lem:cleanup}
Assume that
\[
m'=n^{1.1}
\qquad\text{and}\qquad
\eta \le \frac{1}{4k}.
\]
Then, for all sufficiently large $n$,
\[
\Pr\!\left[\widetilde x_i=x_i\ \text{for all }i\in[n]\right]
\ge 1-n^{-10}.
\]
\end{lemma}

\begin{proof}
We first resolve the global sign ambiguity when $k$ is odd. In that case, for any assignment
$u\in\{\pm1\}^n$, the expected fraction of fresh clean-up equations satisfied by $u$ is strictly
larger than $1/2$ if $u=x$, and strictly smaller than $1/2$ if $u=-x$. Since the clean-up set
contains $m'=n^{1.1}$ independent fresh samples, a standard Chernoff bound implies that, with
probability at least $1-n^{-20}$, the sign check correctly determines whether to keep $\hat x$ or
replace it by $-\hat x$. Thus, after this step, we may assume that
\[
\frac1n\bigl|\{j\in[n]:\hat x_j\neq x_j\}\bigr|\le \eta.
\]

Fix $i\in[n]$, and let
\[
N_i:=\sum_{e\ni i}\mathbf 1_{\{e\text{ is included in the clean-up set}\}}.
\]
Then
\[
N_i\sim \mathrm{Bin}\!\left(\binom{n-1}{k-1},\,p'\right),
\qquad
\mu_i:=\E[N_i]=p'\binom{n-1}{k-1}.
\]
Since
\[
p'=\frac{m'}{\binom{n}{k}},
\]
we have
\[
\mu_i
=
\frac{m'}{\binom{n}{k}}\binom{n-1}{k-1}
=
\frac{k m'}{n}
=
k n^{0.1}.
\]
In particular, $\mu_i$ is polynomially large in $n$. Therefore, by a standard Chernoff bound,
\[
\Pr\!\left[N_i\le \frac{\mu_i}{2}\right]
\le
\exp\!\left(-\frac{\mu_i}{8}\right)
\le
\exp(-\Omega(n^{0.1}))
\le
n^{-20}
\]
for all sufficiently large $n$. Hence, with probability at least $1-n^{-20}$,
\[
N_i\ge \frac12\,k n^{0.1}.
\]

We next analyze the sign of each summand in $T_i$. For any $e\ni i$,
\[
y_e\prod_{j\in e\setminus\{i\}}\hat x_j
=
\xi_e \prod_{j\in e}x_j \prod_{j\in e\setminus\{i\}}\hat x_j
=
\xi_e\,x_i\prod_{j\in e\setminus\{i\}}(x_j\hat x_j).
\]
Let
\[
B:=\{j\in[n]:\hat x_j\neq x_j\},
\qquad |B|\le \eta n.
\]
For a uniformly random $(k-1)$-subset $e\setminus\{i\}$, the probability that it contains at least
one bad coordinate is at most
\[
(k-1)\eta \le k\eta \le \frac14.
\]
Therefore,
\[
\Pr\!\left[\prod_{j\in e\setminus\{i\}}(x_j\hat x_j)=1\right]
\ge \frac34.
\]
It follows that
\[
\Pr\!\left[y_e\prod_{j\in e\setminus\{i\}}\hat x_j=x_i\right]
\ge
\frac34\cdot\frac{1+\delta}{2}
+
\frac14\cdot\frac{1-\delta}{2}
=
\frac12+\frac{\delta}{4}.
\]
Equivalently, after multiplying by $x_i$, each summand is a $\{\pm1\}$-valued random variable
with expectation at least
\[
2\left(\frac12+\frac{\delta}{4}\right)-1
=
\frac{\delta}{2}.
\]

Now condition on the event
\[
N_i\ge \frac12\,k n^{0.1}.
\]
Then $x_iT_i$ is a sum of $N_i$ independent $\{\pm1\}$-valued random variables, each having
mean at least $\delta/2$. Hence
\[
\E[x_iT_i\mid N_i]\ge \frac{\delta}{2}N_i.
\]
By Hoeffding's inequality,
\[
\Pr\!\left[\widetilde x_i\neq x_i \,\middle|\, N_i\right]
=
\Pr\!\left[x_iT_i\le 0 \,\middle|\, N_i\right]
\le
\exp\!\left(
-\frac{(\delta N_i/2)^2}{2N_i}
\right)
=
\exp\!\left(-\frac{\delta^2 N_i}{8}\right).
\]
Using the lower bound on $N_i$, we obtain
\[
\Pr\!\left[\widetilde x_i\neq x_i \,\middle|\, N_i\ge \frac12\,k n^{0.1}\right]
\le
\exp\!\left(-\frac{\delta^2 k}{16}n^{0.1}\right)
\le
n^{-20}
\]
for all sufficiently large $n$ and $\delta^2 n^{0.1}\ge C\log n$.

Combining the two estimates yields
\[
\Pr[\widetilde x_i\neq x_i]
\le
n^{-20}+n^{-20}
\le
n^{-19}
\]
for all sufficiently large $n$. Finally, taking a union bound over all $i\in[n]$, we conclude that
\[
\Pr\!\left[\widetilde x_i=x_i\ \text{for all }i\in[n]\right]
\ge
1-n^{-18}
\ge
1-n^{-10}.
\]
This completes the proof.
\end{proof}

\paragraph{Obtaining an independent clean-up set by sample splitting.}
Suppose the observed equations are generated first, and we want the clean-up stage to be independent of the statistic used in the first stage. This can be achieved by a standard sample-splitting argument.

Specifically, for each observed equation, independently assign it to the clean-up pool with probability $\rho$, and to the main estimation pool with probability $1-\rho$. Then the two pools are independent; if the original inclusion probability is $p$, the clean-up pool has inclusion probability $p'=\rho p$; and the labels in the clean-up pool retain the same noise rate, namely error probability $\frac{1-\delta}{2}$.
Thus we may reserve a small random fraction of the revealed equations for the clean-up step and simply exclude them from the first-stage computation. If $m$ denotes the original sample size and $\rho m=o(m)$, then this changes the effective sample size in the first stage by only a lower-order term.

\begin{remark} We claim that sample splitting is necessary. Indeed, if one were to reuse the same samples that were already used to construct $\hat x$, then $\hat x$ could be correlated with the noise in those samples, and the coordinate-wise majority analysis would no longer be valid. By contrast, sample splitting removes this dependence and allows for a clean direct concentration argument.
\end{remark}

\paragraph{Putting it together}
We now combine the rooted recovery statistic, the sample-splitting argument, and the clean-up step to obtain the final recovery guarantee.
\begin{theorem}\label{thm:full-recovery}
Assume $r$ is larger than a sufficiently large absolute constant, set
\[
D:=rk,
\]
and let $\ell\ge 2$ satisfy
\[
D\ell<\frac{0.3\log n}{\log\log n}.
\]
If
\[
m
\ge
\frac{4\,\ell^{k} k^{k/2+1}}{k!\,2^{k/2}}
\cdot
\frac{n^{k/2}}{\delta^2\,D^{\,k/2-1}},
\]
then the recovery algorithm described in Algorithm~\ref{alg:full-recovery} succeeds with probability at least
\[
1-O\!\left(\ell^{-kD}\right).
\]
More precisely, when $k$ is odd, the algorithm outputs the planted assignment $x$ exactly; when $k$ is even, it outputs either $x$ or $-x$.
\end{theorem}

In particular, for every sufficiently large constant $D$ divisible by $k$, depending only on $k$, the algorithm achieves weak recovery. Moreover, if either $\ell=\omega(1)$ or $D=\omega(1)$, then the success probability tends to $1$. The algorithm achieves exact recovery with high probability when $k$ is odd, and recovery up to global sign with high probability when $k$ is even.
\section{Hypergraph Color Coding and Statistics Approximation}
In this section, we show how to turn the hypergraph statistics from the previous section into efficient algorithms. The main tool is hypergraph color coding, which allows us to replace the full embedding sums by colorful embedding sums that can be computed via dynamic programming over the bounded-width decomposition of the pattern family.
\subsection{Approximating the detection statistic}
We now describe an efficient approximation to the detection statistic $F_{\mathcal H}(z)$. Variants of the same color-coding idea have also been used in \cite{hopkins2017bayesian,li2025smooth,li2025algorithmic} to obtain algorithms faster than brute-force enumeration. The basic idea is to retain only those embeddings whose vertex images receive distinct colors under a random coloring of the ambient vertex set, which makes the resulting statistic amenable to dynamic programming.

Recall that every pattern $H\in\mathcal H$ satisfies $|V(H)|=v=rk\ell$, $|E(H)|=s=2r\ell$.
Define $\rho:=\frac{v!}{v^v}$.
Let $\tau:[n]\to [v]$
be a random coloring obtained by assigning to each vertex of $[n]$ an independent uniformly random color in $[v]$. For any subset $W\subseteq [n]$, define
\[
\chi_\tau(W):=\mathbf{1}\bigl[\tau|_W \text{ is injective}\bigr].
\]
In particular, whenever $|W|=v$,
\[
\Pr\!\bigl[\chi_\tau(W)=1\bigr]=\rho.
\]
Equivalently, this is the probability that all vertices in $W$ receive distinct colors under $\tau$.

For a labeled pattern $H\in\mathcal H$, define the corresponding colorful embedding sum by
\[
G_H(z,\tau)
:=
\sum_{\phi:V(H)\hookrightarrow [n]}
\chi_\tau\!\bigl(\phi(V(H))\bigr)
\prod_{e\in E(H)} z\!\bigl(\phi(e)\bigr).
\]
Thus $G_H(z,\tau)$ counts only those injective embeddings of $H$ whose images are colorful with respect to $\tau$. See Algorithm~\ref{alg:dp-detec} for the detailed computation.

\begin{algorithm}[htbp]
\LinesNumbered
\DontPrintSemicolon
\caption{Computation of $G_H(z,\tau)$}\label{alg:dp-detec}
\KwIn{
A weighted $k$-uniform host hypergraph $z$ on $[n]$, where each hyperedge $\alpha\in\binom{[n]}{k}$ has weight $z(\alpha)\in\{0,\pm1\}$ (with $z(\alpha)=0$ indicating that $\alpha$ is absent); a coloring $\tau:[n]\to [rk\ell]$; and a pattern $H\in\mathcal H(r,k,\ell)$ together with a fixed decomposition
$H=U_1\cup\cdots\cup U_\ell$
as in Definition~\ref{def:family}.
}
\KwOut{
The colorful contribution $G_H(z,\tau)$.
}

For each block $U_a$, let $L(U_a)=\{u_a,u_{a+1}\}$ denote its two interface vertices, where $u_{\ell+1}=u_1$\;

For every $a\in[\ell]$, every color set $C\subseteq [rk\ell]$ with $|C|=rk+1$, every two distinct colors $c_1,c_2\in C$, and every two distinct vertices $x,y\in[n]$, compute the base table $\Lambda_a(x,y;c_1,c_2;C)$
defined by
\[
\Lambda_a(x,y;c_1,c_2;C)
:=
\sum_{\phi:V(U_a)\hookrightarrow[n]}
\mathbf{1}\!\left[
\begin{array}{l}
\tau(\phi(V(U_a)))=C,\\
\phi(u_a)=x,\ \phi(u_{a+1})=y,\\
\tau(x)=c_1,\ \tau(y)=c_2
\end{array}
\right]
\prod_{e\in E(U_a)} z\!\bigl(\phi(e)\bigr).
\]\;
For every $t=2,3,\dots,\ell$, define the suffix table $Y_t(x,y;c_1,c_2;C)$
to be the total contribution of colorful embeddings of the chain $U_t\cup U_{t+1}\cup\cdots\cup U_\ell$
whose two exposed interface vertices are mapped to $x,y$, with colors $c_1,c_2$, and whose image uses exactly the color set $C$\;

Initialize the dynamic program by setting
\[
Y_\ell(x,y;c_1,c_2;C):=\Lambda_\ell(x,y;c_1,c_2;C)
\]
for all valid parameters\;

For $t=\ell-1,\ell-2,\dots,2$, compute $Y_t$ from $Y_{t+1}$ by gluing $U_t$ to the suffix $U_{t+1}\cup\cdots\cup U_\ell$ along their unique shared interface vertex:
\[
Y_t(x,y;c_1,c_2;C)
=
\sum_{\substack{z\in[n]\\ c\in C}}
\ \sum_{\substack{C_1,C_2\subseteq C\\ C_1\cup C_2=C\\ C_1\cap C_2=\{c\}}}
\Lambda_t(x,z;c_1,c;C_1)\,
Y_{t+1}(z,y;c,c_2;C_2).
\]
The sum ranges over all intermediate vertices $z$, intermediate colors $c$, and all compatible decompositions of the color set $C$ in which the two parts intersect exactly in the shared color $c$\;

Finally, compute $G_H(z,\tau)$ by gluing $U_1$ to the suffix $U_2\cup\cdots\cup U_\ell$ and closing the cycle:
\[
G_H(z,\tau)
=
\sum_{\substack{x,y\in[n]\\ x\neq y}}
\ \sum_{\substack{c_1,c_2\in[rk\ell]\\ c_1\neq c_2}}
\ \sum_{\substack{C_1,C_2\subseteq[rk\ell]\\ C_1\cup C_2=[rk\ell]\\ C_1\cap C_2=\{c_1,c_2\}}}
\Lambda_1(x,y;c_1,c_2;C_1)\,
Y_2(y,x;c_2,c_1;C_2).
\]
Here the shared colors $c_1,c_2$ correspond to the two interface vertices through which $U_1$ is attached to the remainder of the cycle\;

\Return{$G_H(z,\tau)$}\;
\end{algorithm}

\begin{lemma}\label{lem:compute-colorful-H}
For any coloring $\tau:[n]\to [rk\ell]$ and any pattern $H\in\mathcal H$,
the colorful count $G_H(z,\tau)$ can be computed by Algorithm~\ref{alg:dp-detec} in time
\[
n^{rk+2+o(1)}.
\]
\end{lemma}
\begin{proof}
Fix $\tau:[n]\to[rk\ell]$ and $H\in\mathcal H$, and write
\[
H=U_1\cup\cdots\cup U_\ell
\]
as in Algorithm~\ref{alg:dp-detec}, where each $U_t\in\mathcal U(r,k)$ has
$|V(U_t)|=rk+1$ and boundary vertices $L(U_t)=\{v_t,v_{t+1}\}$.

By the recursion in Algorithm~\ref{alg:dp-detec}, the computation of
$G_H(z,\tau)$ reduces to filling the tables
\[
\Lambda_t(x,y;c_1,c_2;C)
\qquad\text{and}\qquad
Y_t(x,y;c_1,c_2;C),
\]
where $x,y\in[n]$, $c_1,c_2\in[rk\ell]$, and $C\subseteq[rk\ell]$.
Since $rk\ell=O(\log n/\log\log n)$, the total number of possible color data
$(c_1,c_2,C)$ is $n^{o(1)}$, so the total number of states is $n^{2+o(1)}$.

For each fixed state, $\Lambda_t(x,y;c_1,c_2;C)$ can be computed by brute-force
enumeration of the remaining $rk-1$ vertices of $U_t$, and hence in time
$O(n^{rk+o(1)})$. Therefore all $\Lambda_t$-tables together take time
$O(n^{rk+2+o(1)})$. Next, each $Y_t(x,y;c_1,c_2;C)$ is obtained from the
recursion by summing over an intermediate vertex $z\in[n]$ and over
$n^{o(1)}$ many color splittings, so each entry takes time $n^{1+o(1)}$, and
hence all $Y_t$-tables together take time $O(n^{3+o(1)})$. Finally, the last
summation producing $G_H(z,\tau)$ has $n^{2+o(1)}$ terms.

Combining these bounds, the total running time is
\[
O(n^{rk+2+o(1)})+O(n^{3+o(1)})+O(n^{2+o(1)})
=
O(n^{rk+2+o(1)}).
\]
This proves the claim.
\end{proof}

Conditioned on the observation $z$, each injective embedding survives with probability exactly $\rho$. Therefore
\[
\E\!\left[G_H(z,\tau)\mid z\right]
=
\rho
\sum_{\phi:V(H)\hookrightarrow [n]}
\prod_{e\in E(H)} z\!\bigl(\phi(e)\bigr).
\]
It follows that $G_H(z,\tau)/\rho$ is an unbiased estimator of the full embedding sum
\[
T_H(z)
:=
\sum_{\phi:V(H)\hookrightarrow [n]}
\prod_{e\in E(H)} z\!\bigl(\phi(e)\bigr).
\]

To reduce the additional variance introduced by color coding, we average over several independent colorings. Let $t:=\left\lceil \frac1\rho \right\rceil$,
and let $\tau_1,\dots,\tau_t$ be independent copies of $\tau$. Define
\[
\widetilde T_H(z)
:=
\frac{1}{t\rho}\sum_{a=1}^t G_H(z,\tau_a).
\]
Then $\widetilde T_H(z)$ is still an unbiased estimator of $T_H(z)$, but with reduced variance.

We then define the color-coded approximation to $F_{\mathcal H}(z)$ by replacing each exact embedding sum with its averaged colorful counterpart:
\[
\widetilde F_{\mathcal H}(z)
:=
\sum_{H\in\mathcal H}\widetilde T_H(z)
=
\sum_{H\in\mathcal H}
\left(
\frac{1}{t\rho}\sum_{a=1}^t G_H(z,\tau_a)
\right).
\]

The advantage of this construction is that $\widetilde F_{\mathcal H}(z)$ has the same expectation as the original statistic $F_{\mathcal H}(z)$, while being computable via dynamic programming over colorful partial embeddings. The next lemma shows that $\widetilde F_{\mathcal H}(z)$ is asymptotically equivalent to $F_{\mathcal H}(z)$ in the relevant regime. See Algorithm~\ref{alg:detec-final} for the detailed computation.
\begin{algorithm}[htbp]
\caption{Computation of $\widetilde F_{\mathcal H}(z)$}
\label{alg:detec-final}
\LinesNumbered
\DontPrintSemicolon

\KwIn{A weighted $k$-uniform host hypergraph $z$ on $[n]$, where each hyperedge $\alpha\in\binom{[n]}{k}$ has weight $z(\alpha)\in\{0,\pm1\}$, with $z(\alpha)=0$ indicating that $\alpha$ is absent; integers $r,k,\ell$.}
\KwOut{The approximate statistic $\widetilde F_{\mathcal H}(z)$.}

Set $v\gets rk\ell$, $\rho\gets v!/v^v$, and $t\gets \lceil 1/\rho\rceil$\;

Enumerate the family $\mathcal H=\mathcal H(r,k,\ell)$\;

\ForEach{$H\in\mathcal H$}{
    Fix a decomposition $H=U_1\cup\cdots\cup U_\ell$ such that $U_a\in\mathcal U(r,k)$ for all $a\in[\ell]$\;
}

\For{$a=1$ \KwTo $t$}{
    Sample an independent uniformly random coloring $\tau_a:[n]\to[v]$\;
    \ForEach{$H\in\mathcal H$}{
        Compute $G_H(z,\tau_a)$ using Algorithm~\ref{alg:dp-detec}\;
    }
}

\ForEach{$H\in\mathcal H$}{
    Set $\widetilde T_H(z)\gets \frac{1}{t\rho}\sum_{a=1}^t G_H(z,\tau_a)$\;
}

Set $\widetilde F_{\mathcal H}(z)\gets \sum_{H\in\mathcal H}\widetilde T_H(z)$\;

\Return{$\widetilde F_{\mathcal H}(z)$}\;
\end{algorithm}

\begin{lemma}\label{prop:compute-approx-F}
Under the same parameter regime assumed in the analysis of
$F_{\mathcal H}(z)$, the above algorithm computes
$\widetilde F_{\mathcal H}(z)$ in time
\[
n^{rk+2+o(1)}.
\]
\end{lemma}

\begin{proof}
The algorithm enumerates all $H\in\mathcal H$, and for each $H$, computes
$G_H(z,\tau_a)$ for $a=1,\dots,t$. By the parameter regime, both
$|\mathcal H|$ and $t$ are $n^{o(1)}$. Moreover, by
Lemma~\ref{lem:compute-colorful-H}, for each fixed $H\in\mathcal H$ and each fixed
coloring $\tau_a$, the quantity $G_H(z,\tau_a)$ can be computed in time
\[
O\!\left(n^{rk+2+o(1)}\right).
\]
Hence the total time for all these computations is
\[
|\mathcal H|\cdot t\cdot n^{rk+2+o(1)}
=
n^{o(1)}\cdot n^{rk+2+o(1)}
=
n^{rk+2+o(1)}.
\]
Since the remaining averaging and normalization steps contribute only an
additional $n^{o(1)}$ overhead, the whole algorithm computes
$\widetilde F_{\mathcal H}(z)$ within time $n^{rk+2+o(1)}$.
\end{proof}
We next show that the color-coded statistic is not only efficiently computable, but also statistically faithful to the original statistic.

\begin{lemma}\label{lem:concentration-detec}
If $\widetilde F_{\mathcal H}(z)$ is defined as above, then under the same assumptions as in
the second-moment analysis,
\[
\frac{\widetilde F_{\mathcal H}(z)-F_{\mathcal H}(z)}
{\mathbb E_{\PP_x}[F_{\mathcal H}(z)]}
\;\xrightarrow[L_2]{}\;0
\]
under both $\PP_x$ and $\mathbb Q$.
Hence $\widetilde F_{\mathcal H}(z)$ has the same asymptotic detection power as
$F_{\mathcal H}(z)$.
\end{lemma}
The proof is deferred to Appendix~\ref{apx:proof-concentration1}.

Since convergence in $L^2$ implies convergence in probability, for every fixed $\varepsilon>0$,
\[
\PP_x\!\left(
\left|
\widetilde F_{\mathcal H}(z)-F_{\mathcal H}(z)
\right|
>
\varepsilon\,\E_{\PP_x}[F_{\mathcal H}(z)]
\right)\to 0,
\]
and
\[
\QQ\!\left(
\left|
\widetilde F_{\mathcal H}(z)-F_{\mathcal H}(z)
\right|
>
\varepsilon\,\E_{\PP_x}[F_{\mathcal H}(z)]
\right)\to 0.
\]
Hence $\widetilde F_{\mathcal H}(z)$ has the same asymptotic detection power as $F_{\mathcal H}(z)$.

\begin{remark}
The color-coding step can also be derandomized by replacing the random colorings with a balanced family of perfect hash functions from \cite{alon2010balanced}. Concretely, one may use a $\delta$-balanced $(n,v)$-family so that every $v$-subset of $[n]$ is colorful under essentially the same number of colorings, up to a multiplicative $\delta$ factor. Averaging over such a family preserves the approximation guarantees up to negligible losses, while yielding a deterministic implementation of the statistics.
\end{remark}
\subsection{Approximating the recovery statistic}

We now describe an efficient approximation to the rooted recovery statistic
$F_{\mathcal J,a,b}(z)$. As in the detection case, we use hypergraph color coding to
restrict attention to colorful embeddings, which can then be computed by dynamic
programming over the chain decomposition
\[
J = U_1 \cup \cdots \cup U_\ell
\]
of a pattern $J\in\mathcal J(r,k,\ell)$. This is the path analogue of the cycle-based
dynamic program used for detection. See also Li's recovery-side construction, which uses
the same gluing principle for a chain with two exposed leaves. 

Recall that every $J\in\mathcal J$ has
\[
|V(J)|=v:=rk\ell+1,
\qquad
|E(J)|=s:=2r\ell,
\qquad
L(J)=\{v_1,v_{\ell+1}\}.
\]
Define $\rho_{\rm rec}:=\frac{v!}{v^v}$.
Let $\tau:[n]\to [v]$
be a random coloring obtained by assigning to each vertex of $[n]$ an independent
uniform color in $[v]$. For any subset $W\subseteq [n]$, define
\[
\chi_\tau(W):=\mathbf 1\!\bigl[\tau|_W \text{ is injective}\bigr].
\]
Then for every $W\subseteq[n]$ with $|W|=v$,
\[
\Pr\!\bigl[\chi_\tau(W)=1\bigr]=\rho_{\rm rec}.
\]

For each labeled pattern $J\in\mathcal J$ and each ordered pair $(a,b)\in[n]^2$ with
$a\neq b$, define the colorful rooted embedding sum
\[
G_{J,a,b}(z,\tau)
:=
\sum_{\substack{\phi:V(J)\hookrightarrow[n]\\ \phi(v_1)=a,\ \phi(v_{\ell+1})=b}}
\chi_\tau\!\bigl(\phi(V(J))\bigr)
\prod_{e\in E(J)} z\!\bigl(\phi(e)\bigr).
\]
Thus $G_{J,a,b}(z,\tau)$ counts only those injective rooted embeddings of $J$ whose
image is colorful under $\tau$. See Algorithm~\ref{alg:dp-rec} for the detailed computation of $G_{J,a,b}(z,\tau)$.

\begin{algorithm}[htbp]\label{alg:dp-rec}
\LinesNumbered
\DontPrintSemicolon
\caption{Computation of $G_{J,a,b}(z,\tau)$}
\KwIn{
A weighted $k$-uniform host hypergraph $z$ on $[n]$, where each hyperedge
$\alpha\in\binom{[n]}{k}$ has weight $z(\alpha)\in\{0,\pm1\}$; a coloring
$\tau:[n]\to [v]$, where $v=rk\ell+1$; a pattern $J\in\mathcal J(r,k,\ell)$
together with a fixed decomposition
\[
J=U_1\cup\cdots\cup U_\ell
\]
as in Definition~\ref{def:family-recovery}; and a prescribed ordered pair
$(a,b)\in[n]^2$ with $a\neq b$.
}
\KwOut{
The colorful rooted contribution $G_{J,a,b}(z,\tau)$.
}

For each block $U_t$, let $L(U_t)=\{v_t,v_{t+1}\}$ denote its two interface vertices\;

For every $t\in[\ell]$, every color set $C\subseteq [v]$ with $|C|=rk+1$, every two distinct
colors $c_1,c_2\in C$, and every two distinct vertices $x,y\in[n]$, compute the base table
\[
\Lambda_t(x,y;c_1,c_2;C)
:=
\sum_{\phi:V(U_t)\hookrightarrow[n]}
\mathbf 1\!\left[
\begin{array}{l}
\tau(\phi(V(U_t)))=C,\\
\phi(v_t)=x,\ \phi(v_{t+1})=y,\\
\tau(x)=c_1,\ \tau(y)=c_2
\end{array}
\right]
\prod_{e\in E(U_t)} z\!\bigl(\phi(e)\bigr).
\]
\;

For every $t=1,2,\dots,\ell$, define the suffix table $Y_t(x,y;c_1,c_2;C)$ to be the total
contribution of colorful embeddings of the chain
\[
U_t\cup U_{t+1}\cup\cdots\cup U_\ell
\]
whose two exposed interface vertices are mapped to $x,y$, with colors $c_1,c_2$, and whose
image uses exactly the color set $C$\;

Initialize
\[
Y_\ell(x,y;c_1,c_2;C):=\Lambda_\ell(x,y;c_1,c_2;C)
\]
for all valid parameters\;

For $t=\ell-1,\ell-2,\dots,1$, compute $Y_t$ from $Y_{t+1}$ by gluing $U_t$ to the suffix
$U_{t+1}\cup\cdots\cup U_\ell$ along the shared interface vertex $v_{t+1}$:
\[
Y_t(x,y;c_1,c_2;C)
=
\sum_{\substack{z\in[n]\\ c\in C}}
\ \sum_{\substack{C_1,C_2\subseteq C\\ C_1\cup C_2=C\\ C_1\cap C_2=\{c\}}}
\Lambda_t(x,z;c_1,c;C_1)\,
Y_{t+1}(z,y;c,c_2;C_2).
\]
\;

Return
\[
G_{J,a,b}(z,\tau)
=
\sum_{\substack{c_1,c_2\in[v]\\ c_1\neq c_2}}
Y_1(a,b;c_1,c_2;[v]).
\]
\;
\end{algorithm}

Conditioned on $z$, every rooted injective embedding survives with probability exactly
$\rho_{\rm rec}$. Therefore
\[
\E\!\left[G_{J,a,b}(z,\tau)\mid z\right]
=
\rho_{\rm rec}\, T_{J,a,b}(z),
\]
where
\[
T_{J,a,b}(z)
=
\sum_{\substack{\phi:V(J)\hookrightarrow[n]\\ \phi(v_1)=a,\ \phi(v_{\ell+1})=b}}
\prod_{e\in E(J)} z\!\bigl(\phi(e)\bigr).
\]
Hence $G_{J,a,b}(z,\tau)/\rho_{\rm rec}$ is an unbiased estimator of $T_{J,a,b}(z)$.

To reduce the additional variance coming from color coding, let $t:=\left\lceil \frac{1}{\rho_{\rm rec}}\right\rceil$,
and let $\tau_1,\dots,\tau_t$ be independent copies of $\tau$. Define
\[
\widetilde T_{J,a,b}(z)
:=
\frac{1}{t\rho_{\rm rec}}
\sum_{q=1}^t G_{J,a,b}(z,\tau_q).
\]
We then set
\[
\widetilde F_{\mathcal J,a,b}(z)
:=
\sum_{J\in\mathcal J} \widetilde T_{J,a,b}(z).
\]
By construction,
\[
\E\!\left[\widetilde T_{J,a,b}(z)\mid z\right]=T_{J,a,b}(z),
\qquad
\E\!\left[\widetilde F_{\mathcal J,a,b}(z)\mid z\right]=F_{\mathcal J,a,b}(z).
\]
See Algorithm~\ref{alg:dp-rec} for the detailed computation of $G_{J,a,b}(z,\tau)$, and Algorithm~\ref{alg:rec-final} for the computation of $\widetilde F_{\mathcal J,a,b}(z)$.

\begin{algorithm}[htbp]
\caption{Computation of $\widetilde F_{\mathcal J,a,b}(z)$}
\label{alg:rec-final}
\LinesNumbered
\DontPrintSemicolon

\KwIn{A weighted $k$-uniform host hypergraph $z$ on $[n]$, where each hyperedge $\alpha\in\binom{[n]}{k}$ has weight $z(\alpha)\in\{0,\pm1\}$; integers $r,k,\ell$; and an ordered pair $(a,b)\in[n]^2$ with $a\neq b$.}
\KwOut{The approximate rooted statistic $\widetilde F_{\mathcal J,a,b}(z)$.}

Set $v\gets rk\ell+1$, $\rho_{\mathrm{rec}}\gets v!/v^v$, and $t\gets \lceil 1/\rho_{\mathrm{rec}}\rceil$\;

Enumerate the family $\mathcal J=\mathcal J(r,k,\ell)$\;

\ForEach{$J\in\mathcal J$}{
    Fix a decomposition $J=U_1\cup\cdots\cup U_\ell$ such that $U_t\in\mathcal U(r,k)$ for all $t\in[\ell]$\;
}

\For{$q=1$ \KwTo $t$}{
    Sample an independent uniformly random coloring $\tau_q:[n]\to[v]$\;
    \ForEach{$J\in\mathcal J$}{
        Compute $G_{J,a,b}(z,\tau_q)$ using Algorithm~\ref{alg:dp-rec}\;
    }
}

\ForEach{$J\in\mathcal J$}{
    Set $\widetilde T_{J,a,b}(z)\gets \frac{1}{t\rho_{\mathrm{rec}}}\sum_{q=1}^t G_{J,a,b}(z,\tau_q)$\;
}

Set $\widetilde F_{\mathcal J,a,b}(z)\gets \sum_{J\in\mathcal J}\widetilde T_{J,a,b}(z)$\;

\Return{$\widetilde F_{\mathcal J,a,b}(z)$}\;
\end{algorithm}
\begin{lemma}\label{lem:compute-colorful-J}
For any coloring $\tau:[n]\to[rk\ell+1]$, any ordered pair $(a,b)\in[n]^2$ with $a\neq b$,
and any pattern $J\in\mathcal J(r,k,\ell)$, the colorful rooted count $G_{J,a,b}(z,\tau)$ can be
computed by Algorithm~\ref{alg:dp-rec} in time
\[
n^{rk+2+o(1)}.
\]
\end{lemma}

\begin{proof}
Write
\[
J=U_1\cup\cdots\cup U_\ell
\]
as in Algorithm~\ref{alg:dp-rec}, where $U_1,\dots,U_\ell\in\mathcal U(r,k)$ form a chain whose two leaves are mapped to the prescribed vertices $a,b\in[n]$. As in the proof of Lemma~\ref{lem:compute-colorful-H}, the computation of
$G_{J,a,b}(z,\tau)$ is reduced to filling dynamic-programming tables indexed by two vertices
of $[n]$ together with a color profile. Since $rk\ell=o(\log n/\log\log n)$, the total number
of possible color profiles is $n^{o(1)}$, and hence the total number of states is $n^{2+o(1)}$.

For each fixed state, the corresponding local quantity is obtained by brute-force enumeration of
the remaining $rk-1$ vertices inside one block $U_t$, which takes time $n^{rk+o(1)}$. Thus
all base tables can be filled in time $n^{rk+2+o(1)}$. The remaining suffix tables are computed
from the recursion by summing over one intermediate vertex and over $n^{o(1)}$ many color
splittings, so all such tables together require time $n^{3+o(1)}$. The final summation then
contributes only $n^{2+o(1)}$.

Combining these bounds gives total running time
\[
n^{rk+2+o(1)}+n^{3+o(1)}+n^{2+o(1)}
=
O\!\left(n^{rk+2+o(1)}\right).
\]
This proves the claim.
\end{proof}

\begin{lemma}\label{lem:compute-approx-F-rec}
Under the same parameter regime assumed in the analysis of $F_{\mathcal J,a,b}(z)$, the above
algorithm computes $\widetilde F_{\mathcal J,a,b}(z)$ in time
\[
n^{rk+2+o(1)}.
\]
More precisely, if $r,k,\ell$ are treated as fixed constants, then the running time is
\[
O\!\left(n^{rk+2+o(1)}\right).
\]
\end{lemma}

\begin{proof}
The proof is the same as that of Proposition~\ref{prop:compute-approx-F}. Indeed, the algorithm
enumerates all $J\in\mathcal J(r,k,\ell)$, and for each $J$, computes
$G_{J,a,b}(z,\tau_s)$ for all sampled colorings $\tau_s$. Under the parameter regime,
both $|\mathcal J|$ and the number of sampled colorings are $n^{o(1)}$. By
Lemma~\ref{lem:compute-colorful-J}, each quantity $G_{J,a,b}(z,\tau_s)$ can be computed in time
\[
O\!\left(n^{rk+2+o(1)}\right).
\]
Therefore the total running time is
\[
|\mathcal J|\cdot n^{o(1)}\cdot n^{rk+2+o(1)}
=
n^{rk+2+o(1)}.
\]
If $r,k,\ell$ are fixed constants, this is
\[
O\!\left(n^{rk+2+o(1)}\right).
\]
Since the remaining averaging and normalization steps contribute only $n^{o(1)}$ overhead, the
whole algorithm computes $\widetilde F_{\mathcal J,a,b}(z)$ within the claimed time bound.
\end{proof}
\begin{lemma}\label{lem:concentration-rec}
If $\widetilde F_{\mathcal J,a,b}(z)$ is defined as above, then under the same assumptions as in
the second-moment analysis of $F_{\mathcal J,a,b}(z)$,
\[
\frac{\widetilde F_{\mathcal J,a,b}(z)-F_{\mathcal J,a,b}(z)}
{\E_{\PP_x}[F_{\mathcal J,a,b}(z)]}
\;\xrightarrow[L_2]{}\;0
\]
under the planted model. Hence $\widetilde F_{\mathcal J,a,b}(z)$ has the same asymptotic recovery
power as $F_{\mathcal J,a,b}(z)$.
\end{lemma}

\begin{proof}
The proof is entirely analogous to that of the corresponding concentration result for
$\widetilde F_{\mathcal H}(z)$. One applies the same color-coding approximation argument and uses
the same second-moment bounds established for $F_{\mathcal J,a,b}(z)$ under the planted model.
Since the required estimates are identical up to replacing the cycle family $\mathcal H$ by the
rooted path family $\mathcal J$, we omit the details.
\end{proof}

\subsection{Putting it together}

Combining the statistical recovery guarantee from Theorem~\ref{thm:full-recovery} with the color-coding runtime bound in Lemma~\ref{lem:compute-approx-F-rec} and the approximation guarantee in Lemma~\ref{lem:concentration-rec}, we obtain the following efficient recovery theorem. Since the recovery procedure evaluates the rooted statistic for all $b\in[n]$ with a fixed anchor vertex, the overall running time gains an additional factor of $n$.

\begin{theorem}\label{thm:full-recovery-formal}
Fix $k\ge 3$. There exists a sufficiently large absolute constant $C>0$ such that the following holds. Let $D$ be a multiple of $k$, and let $\ell\ge 2$ satisfy
\[
D\ell<\frac{0.3\log n}{\log\log n}.
\]
Then Algorithm~\ref{alg:full-recovery} runs in time
\[
n^{D+3+o(1)}
\]
and, in the noisy $k$-XOR model, recovers the planted assignment up to global sign, and in fact exactly when $k$ is odd, provided that
\[
m
\ge
\frac{4\,\ell^{k} k^{k/2+1}}{k!\,2^{k/2}}
\cdot
\frac{n^{k/2}}{\delta^2\,D^{\,k/2-1}}.
\]
Moreover, the success probability is at least
\[
1-O\!\left(\ell^{-kD}\right).
\]
In particular, the algorithm achieves weak recovery. If in addition $D=\omega(1)$, or if $\ell = \omega(1)$,
then the success probability improves to $1-o(1)$. More precisely, when $k$ is odd, the algorithm outputs the planted assignment $x$ exactly; when $k$ is even, it outputs either $x$ or $-x$.
\end{theorem}
\section{Low-Degree Hardness of Noisy $k$-XOR}
In this section, we investigate the low-degree hardness of the noisy $k$-XOR problem, thereby providing additional evidence that our algorithm is essentially tight also with respect to the sample complexity. 
In our setting, the observation is denoted by $z=(z_\alpha)_{\alpha\in\binom{[n]}{k}}$, where each coordinate takes values in $\{0,\pm 1\}$.
Accordingly, we view a polynomial test as a multilinear polynomial in the coordinate functions $(z_\alpha)$, and we write $\mathbb{R}[z]_{\le D}$
for the space of all real-valued polynomials in these variables of total degree at most $D$.
Since each coordinate satisfies $z_\alpha^3=z_\alpha$,
every polynomial may be reduced to one in which each variable appears with exponent at most $2$.
Thus, in this model, a degree-$D$ polynomial means a sequence $f=f_n\in \mathbb{R}[z]_{\le D_n}$,
indexed by the problem size $n$, where each monomial has total degree at most $D_n$ and degree at most $2$ in each individual variable.
The low-degree heuristic predicts that if all such polynomials fail to distinguish $\PP$ from $\QQ$, then no polynomial-time algorithm should succeed.
Motivated by this principle, we study the degree-$D$ projection of the likelihood ratio and the associated low-degree norm.

Let $L(z):=\frac{\mathrm d\PP}{\mathrm d\QQ}(z)$ denote the likelihood ratio, and let $L_{\le D}$ be its orthogonal projection onto the space of degree-$\le D$ polynomials in $L^2(\QQ)$. Then
\[
\sup_{\deg(f)\le D}
\frac{\E_{\PP}[f]-\E_{\QQ}[f]}{\sqrt{\Var_{\QQ}(f)}}
=
\|L_{\le D}-1\|_{L^2(\QQ)}.
\]
This quantity captures the optimal distinguishing power of degree-$D$ polynomials between $\PP$ and $\QQ$, and is therefore a central object in the low-degree framework. The low-degree heuristic predicts that if $\chi_{\le D}^2(\PP,\QQ)=O(1)$,
then no algorithm running in time $\exp(\widetilde{\Theta}(D))$ can strongly distinguish $\PP$ from $\QQ$.

\paragraph{Low-degree likelihood ratio for noisy $k$-XOR.}
Let $\mathcal M:=\binom{[n]}{k}$, $M:=|\mathcal M|=\binom{n}{k}$.
Under the null model $\QQ$, the coordinates $(z_\alpha)_{\alpha\in\mathcal M}$ are independent with
\[
\QQ(z_\alpha=0)=1-p,
\qquad
\QQ(z_\alpha=\pm1)=\frac p2.
\]
Under the planted model $\PP$, we first sample $x\in\{\pm1\}^n$ uniformly, and then independently for each $\alpha\in\mathcal M$,
\[
z_\alpha=b_\alpha x_\alpha \xi_\alpha,
\qquad
x_\alpha:=\prod_{i\in\alpha}x_i,
\]
where $b_\alpha\sim\mathrm{Ber}(p)$ and $\xi_\alpha\sim\mathrm{Rad}(\delta)$. For fixed $x$, let $\PP_x$ denote the conditional planted law. Then for every $\alpha\in\mathcal M$,
\[
\frac{\mathrm{d}\PP_x}{\mathrm{d}\QQ}(z_\alpha)=1+\delta x_\alpha z_\alpha,
\]
and therefore the full likelihood ratio is
\[
L_x(z)=\prod_{\alpha\in\mathcal M}\bigl(1+\delta x_\alpha z_\alpha\bigr).
\]
To expand $L_x$, define $\phi_\alpha(z):=\frac{z_\alpha}{\sqrt p}$, and $\Phi_S(z):=\prod_{\alpha\in S}\phi_\alpha(z)$  for any $S\subseteq\mathcal M$.
Since under $\QQ$ the variables $z_\alpha$ are independent with $\E_{\QQ}[z_\alpha]=0$ and $\E_{\QQ}[z_\alpha^2]=p$, the family $\{\Phi_S\}_{S\subseteq\mathcal M}$ is an orthonormal basis of $L^2(\QQ)$. Hence
\[
L_x(z)
=
\sum_{S\subseteq\mathcal M}
(\delta\sqrt p)^{|S|}
\Bigl(\prod_{\alpha\in S}x_\alpha\Bigr)\Phi_S(z).
\]
Averaging over $x$, we obtain
\[
L(z):=\E_x[L_x(z)]
=
\sum_{S\subseteq\mathcal M}
(\delta\sqrt p)^{|S|}
\E_x\!\Bigl[\prod_{\alpha\in S}x_\alpha\Bigr]\Phi_S(z).
\]
Now we have $\prod_{\alpha\in S}x_\alpha=\prod_{i=1}^n x_i^{d_S(i)}$,
where $d_S(i)$ is the degree of vertex $i$ in the $k$-uniform hypergraph $S$. Since the coordinates of $x$ are i.i.d. uniform signs,
$\E_x\!\Bigl[\prod_{\alpha\in S}x_\alpha\Bigr]
=
\mathbf 1_{\{S\text{ even}\}}$,
where $S$ is called \emph{even} if every vertex has even degree in $S$. Therefore
\[
L(z)
=
\sum_{\substack{S\subseteq\mathcal M\\ S\text{ even}}}
(\delta\sqrt p)^{|S|}\Phi_S(z),
\]
so that
\[
\widehat L(S)=
(\delta\sqrt p)^{|S|}\mathbf 1_{\{S\text{ even}\}}.
\]
It follows that the degree-$D$ truncation satisfies
\[
L_{\le D}(z)
=
\sum_{\substack{S\subseteq\mathcal M\\ |S|\le D\\ S\text{ even}}}
(\delta\sqrt p)^{|S|}\Phi_S(z),
\]
and hence
\begin{equation}\label{eq:LD-norm-even-hypergraphs}
\|L_{\le D}\|_{L^2(\QQ)}^2
=
\sum_{\substack{S\subseteq\mathcal M\\ |S|\le D\\ S\text{ even}}}
(p\delta^2)^{|S|}
=
\sum_{t=0}^D N_t (p\delta^2)^t,
\end{equation}
where $N_t$ denotes the number of even $k$-uniform hypergraphs on $[n]$ with exactly $t$ edges. Thus the low-degree norm is determined by the weighted count of even $k$-uniform subhypergraphs.

\begin{lemma}[Counting bound for even $k$-uniform hypergraphs]\label{lem:counting-even}
Let $N_t$ denote the number of simple even $k$-uniform hypergraphs on $[n]$ with exactly $t$ edges. Then we have
\[
N_t
\le
\left(
e^{k/2+2}\frac{k^{k/2+1/2}}{k!}\,
n^{k/2}\,
t^{k/2-1}
\right)^t
\qquad\text{for all } 1\le t\le \frac{n}{k}.
\]
\end{lemma}
The proof is deferred to Appendix~\ref{post:counting}.

\begin{corollary}
Let
\[
C_k:=e^{k/2+2}\frac{k^{k/2+1/2}}{k!}.
\]
Then
\[
\|L_{\le D}\|_{L^2(\QQ)}^2
\le
1+\sum_{t=1}^D
\bigl(C_k\, n^{k/2}\, t^{k/2-1}\, p\delta^2\bigr)^t.
\]
In particular, if
$C_k\, n^{k/2}\, D^{k/2-1}\, p\delta^2<1$,
then
\[
\|L_{\le D}\|_{L^2(\QQ)}^2=O(1).
\]
Moreover, if $C_k\, n^{k/2}\, D^{k/2-1}\, p\delta^2=o(1)$,
then
\[
\|L_{\le D}\|_{L^2(\QQ)}^2=1+o(1).
\]
\end{corollary}
\begin{proof}
By the previous lemma,
\[
N_t\le \bigl(C_k\, n^{k/2}\, t^{k/2-1}\bigr)^t,
\qquad
C_k=e^{k/2+2}\frac{k^{k/2+1/2}}{k!}.
\]
Therefore
\[
\|L_{\le D}\|_{L^2(\QQ)}^2
=
\sum_{t=0}^D N_t(p\delta^2)^t
\le
1+\sum_{t=1}^D
\bigl(C_k\, n^{k/2}\, t^{k/2-1}\, p\delta^2\bigr)^t.
\]
The remaining conclusions follow by using $t^{k/2-1}\le D^{k/2-1}$ for $t\le D$ and summing the resulting geometric series.
\end{proof}
\paragraph{Interpretation.}
Since
\[
p=\frac{m}{\binom{n}{k}},
\qquad
C_k=e^{k/2+2}\frac{k^{k/2+1/2}}{k!},
\]
the preceding corollary implies that the low-degree norm remains bounded whenever
\[
C_k\,\frac{m}{\binom{n}{k}}\,\delta^2\,D^{\,k/2-1}n^{k/2}<1.
\]
Equivalently, whenever
\[
m<
\frac{\binom{n}{k}}{C_k\,n^{k/2}D^{\,k/2-1}\delta^2}
\sim
\frac{n^{k/2}}{e^{k/2+2}k^{k/2+1/2}D^{\,k/2-1}\delta^2},
\]
the degree-$D$ low-degree norm is $O(1)$. In the low-degree framework, this means that degree-$D$ polynomials do not strongly distinguish $\PP$ from $\QQ$ in this regime. Thus the usual $n^{k/2}$ signal-to-noise threshold persists, up to the additional factor $D^{k/2-1}$ and an explicit $k$-dependent constant.

Moreover, if
\[
m\delta^2=o\!\left(
\frac{n^{k/2}}{e^{k/2+2}k^{k/2+1/2}D^{\,k/2-1}}
\right),
\]
then
\[
\|L_{\le D}\|_{L^2(\QQ)}^2=1+o(1).
\]
In particular, degree-$D$ polynomials have asymptotically no distinguishing power in this regime. Consequently, under the low-degree heuristic, one should not expect any algorithm running in time $n^{O(D)}$ to distinguish $\PP$ from $\QQ$ here.

\begin{theorem}\label{thm:lowdegree-bound}
Fix $k\ge 3$, and let $D=D(n)\ge 1$. Let $C_k:=e^{k/2+2}\frac{k^{k/2+1/2}}{k!}$. If
\[
m\delta^2=o\!\left(\frac{n^{k/2}}{C_k\,D^{\,k/2-1}}\right),
\]
then for every sequence of polynomial tests $f=f_n$ of degree at most $D$,
\[
\frac{\E_{\PP}[f]-\E_{\QQ}[f]}{\sqrt{\Var_{\QQ}(f)}}=o(1),
\]
and no sequence of degree-$D$ polynomial tests can distinguish $\PP$ from $\QQ$ with nonvanishing advantage.
If
\[
m\delta^2 = O\left(\frac{n^{k/2}}{C_k\,D^{\,k/2-1}}\right)
\]
and $D$ is a constant, then
\[
\frac{\E_{\PP}[f]-\E_{\QQ}[f]}{\sqrt{\Var_{\QQ}(f)}}=O(1),
\]
Consequently, degree-$D$ polynomial tests do not strongly distinguish $\PP$ from $\QQ$.
\end{theorem}
\section{Discussion and Open Questions}

We conclude with several directions for future work.

First, variants of the ideas developed here have already proved useful in several related inference problems, including stochastic block models \cite{hopkins2017bayesian,mao2024testing,chen2026detecting}, tensor PCA \cite{li2025smooth}, spiked models \cite{li2025algorithmic}. This raises the broader question of whether similar hypergraph-based statistics, combined with color coding and bounded-width dynamic programming, can be extended to other average-case problems, such as more general planted $k$-CSPs. Such extensions would likely require new gadget families together with sharper combinatorial and probabilistic estimates.

Second, our observation model takes values in $\{0,\pm1\}$, rather than the more standard fully observed $\{\pm1\}$-valued setting, yet it still exhibits a matching low-degree barrier. This suggests that the low-degree framework should be developed further to more systematically capture inference problems with partial observations, nonbinary alphabets, and more general sampling mechanisms.
\section*{Acknowledgments}

The author would like to thank Xin Li and Zhaienhe Zhou for helpful discussions on LPN problems. The author is also grateful to Tim Kunisky for valuable discussions on the low-degree polynomial framework and color-coding. Many of the ideas in this paper grew out of earlier joint work with Tim.
\bibliographystyle{alpha}
\bibliography{references}

@article{li2025smooth,
  title={A Smooth Computational Transition in Tensor PCA},
  author={Li, Zhangsong},
  journal={arXiv preprint arXiv:2509.09904},
  year={2025}
}

@article{alon1995color,
  title={Color-coding},
  author={Alon, Noga and Yuster, Raphael and Zwick, Uri},
  journal={Journal of the ACM (JACM)},
  volume={42},
  number={4},
  pages={844--856},
  year={1995},
  publisher={ACM New York, NY, USA}
}

@article{alon2010balanced,
  title={Balanced families of perfect hash functions and their applications},
  author={Alon, Noga and Gutner, Shai},
  journal={ACM Transactions on Algorithms (TALG)},
  volume={6},
  number={3},
  pages={1--12},
  year={2010},
  publisher={ACM New York, NY, USA}
}

@article{barak2019nearly,
  title={A nearly tight sum-of-squares lower bound for the planted clique problem},
  author={Barak, Boaz and Hopkins, Samuel and Kelner, Jonathan and Kothari, Pravesh K and Moitra, Ankur and Potechin, Aaron},
  journal={SIAM Journal on Computing},
  volume={48},
  number={2},
  pages={687--735},
  year={2019},
  publisher={SIAM}
}

@article{hopkins2017bayesian,
  title={Bayesian estimation from few samples: community detection and related problems},
  author={Hopkins, Samuel B and Steurer, David},
  journal={arXiv preprint arXiv:1710.00264},
  year={2017}
}

@book{hopkins2018statistical,
  title={Statistical inference and the sum of squares method},
  author={Hopkins, Samuel},
  year={2018},
  publisher={Cornell University}
}

@inproceedings{kunisky2019notes,
  title={Notes on computational hardness of hypothesis testing: Predictions using the low-degree likelihood ratio},
  author={Kunisky, Dmitriy and Wein, Alexander S and Bandeira, Afonso S},
  booktitle={ISAAC Congress (International Society for Analysis, its Applications and Computation)},
  pages={1--50},
  year={2019},
  organization={Springer}
}

@article{basu2025solving,
  title   = {Solving Random Planted {CSPs} below the {$n^{k/2}$} Threshold},
  author  = {Basu, Arpon and Hsieh, Jun-Ting and Lin, Andrew D. and Manohar, Peter},
  journal = {arXiv preprint arXiv:2507.10833},
  year    = {2025},
  doi     = {10.48550/arXiv.2507.10833},
}

@article{feldman2015subsampled,
  title={Subsampled power iteration: a unified algorithm for block models and planted csp's},
  author={Feldman, Vitaly and Perkins, Will and Vempala, Santosh},
  journal={Advances in Neural Information Processing Systems},
  volume={28},
  year={2015}
}

@inproceedings{applebaum2010public,
  title={Public-key cryptography from different assumptions},
  author={Applebaum, Benny and Barak, Boaz and Wigderson, Avi},
  booktitle={Proceedings of the forty-second ACM symposium on Theory of computing},
  pages={171--180},
  year={2010}
}

@inproceedings{bogdanov2019xor,
  title={XOR codes and sparse learning parity with noise},
  author={Bogdanov, Andrej and Sabin, Manuel and Vasudevan, Prashant Nalini},
  booktitle={Proceedings of the Thirtieth Annual ACM-SIAM Symposium on Discrete Algorithms},
  pages={986--1004},
  year={2019},
  organization={SIAM}
}

@inproceedings{applebaum2012pseudorandom,
  title={Pseudorandom generators with long stretch and low locality from random local one-way functions},
  author={Applebaum, Benny},
  booktitle={Proceedings of the forty-fourth annual ACM symposium on Theory of computing},
  pages={805--816},
  year={2012}
}

@inproceedings{guruswami2022algorithms,
  title={Algorithms and certificates for Boolean CSP refutation: smoothed is no harder than random},
  author={Guruswami, Venkatesan and Kothari, Pravesh K and Manohar, Peter},
  booktitle={Proceedings of the 54th Annual ACM SIGACT Symposium on Theory of Computing},
  pages={678--689},
  year={2022}
}

@inproceedings{hsieh2023simple,
  title={A simple and sharper proof of the hypergraph Moore bound},
  author={Hsieh, Jun-Ting and Kothari, Pravesh K and Mohanty, Sidhanth},
  booktitle={Proceedings of the 2023 Annual ACM-SIAM Symposium on Discrete Algorithms (SODA)},
  pages={2324--2344},
  year={2023},
  organization={SIAM}
}

@inproceedings{raghavendra2017strongly,
  title={Strongly refuting random CSPs below the spectral threshold},
  author={Raghavendra, Prasad and Rao, Satish and Schramm, Tselil},
  booktitle={Proceedings of the 49th Annual ACM SIGACT Symposium on Theory of Computing},
  pages={121--131},
  year={2017}
}

@inproceedings{alekhnovich2003more,
  title={More on average case vs approximation complexity},
  author={Alekhnovich, Michael},
  booktitle={44th Annual IEEE Symposium on Foundations of Computer Science, 2003. Proceedings.},
  pages={298--307},
  year={2003},
  organization={IEEE}
}

@inproceedings{applebaum2017secure,
  title={Secure arithmetic computation with constant computational overhead},
  author={Applebaum, Benny and Damg{\aa}rd, Ivan and Ishai, Yuval and Nielsen, Michael and Zichron, Lior},
  booktitle={Annual International Cryptology Conference},
  pages={223--254},
  year={2017},
  organization={Springer}
}

@inproceedings{couteau2021silver,
  title={Silver: Silent VOLE and oblivious transfer from hardness of decoding structured LDPC codes},
  author={Couteau, Geoffroy and Rindal, Peter and Raghuraman, Srinivasan},
  booktitle={Annual International Cryptology Conference},
  pages={502--534},
  year={2021},
  organization={Springer}
}

@inproceedings{raghuraman2023expand,
  title={Expand-convolute codes for pseudorandom correlation generators from LPN},
  author={Raghuraman, Srinivasan and Rindal, Peter and Tanguy, Titouan},
  booktitle={Annual International Cryptology Conference},
  pages={602--632},
  year={2023},
  organization={Springer}
}

@inproceedings{ragavan2025indistinguishability,
  title={Indistinguishability obfuscation from bilinear maps and LPN variants},
  author={Ragavan, Seyoon and Vafa, Neekon and Vaikuntanathan, Vinod},
  booktitle={Theory of Cryptography Conference},
  pages={3--36},
  year={2025},
  organization={Springer}
}

@inproceedings{dao2023multi,
  title={Multi-party homomorphic secret sharing and sublinear MPC from sparse LPN},
  author={Dao, Quang and Ishai, Yuval and Jain, Aayush and Lin, Huijia},
  booktitle={Annual International Cryptology Conference},
  pages={315--348},
  year={2023},
  organization={Springer}
}

@article{applebaum2016cryptographic,
  title={Cryptographic hardness of random local functions: Survey},
  author={Applebaum, Benny},
  journal={Computational complexity},
  volume={25},
  number={3},
  pages={667--722},
  year={2016},
  publisher={Springer}
}

@inproceedings{feige2006witnesses,
  title={Witnesses for non-satisfiability of dense random 3CNF formulas},
  author={Feige, Uriel and Kim, Jeong Han and Ofek, Eran},
  booktitle={2006 47th Annual IEEE Symposium on Foundations of Computer Science (FOCS'06)},
  pages={497--508},
  year={2006},
  organization={IEEE}
}

@inproceedings{allen2015refute,
  title={How to refute a random CSP},
  author={Allen, Sarah R and O'Donnell, Ryan and Witmer, David},
  booktitle={2015 IEEE 56th Annual Symposium on Foundations of Computer Science},
  pages={689--708},
  year={2015},
  organization={IEEE}
}

@inproceedings{chen2025algorithms,
  title={Algorithms for Sparse LPN and LSPN Against Low-noise},
  author={Chen, Xue and Shu, Wenxuan and Zhou, Zhaienhe},
  booktitle={The Thirty Eighth Annual Conference on Learning Theory},
  pages={1091--1093},
  year={2025},
  organization={PMLR}
}

@article{bresler2026average,
  title={Average-Case Reductions for $ k $-XOR and Tensor PCA},
  author={Bresler, Guy and Harbuzova, Alina},
  journal={arXiv preprint arXiv:2601.19016},
  year={2026}
}

@inproceedings{hopkins2017efficient,
  title={Efficient bayesian estimation from few samples: community detection and related problems},
  author={Hopkins, Samuel B and Steurer, David},
  booktitle={2017 IEEE 58th Annual Symposium on Foundations of Computer Science (FOCS)},
  pages={379--390},
  year={2017},
  organization={IEEE}
}

@inproceedings{hopkins2017power,
  title={The power of sum-of-squares for detecting hidden structures},
  author={Hopkins, Samuel B and Kothari, Pravesh K and Potechin, Aaron and Raghavendra, Prasad and Schramm, Tselil and Steurer, David},
  booktitle={2017 IEEE 58th Annual Symposium on Foundations of Computer Science (FOCS)},
  pages={720--731},
  year={2017},
  organization={IEEE}
}

@article{blum2003noise,
  title={Noise-tolerant learning, the parity problem, and the statistical query model},
  author={Blum, Avrim and Kalai, Adam and Wasserman, Hal},
  journal={Journal of the ACM (JACM)},
  volume={50},
  number={4},
  pages={506--519},
  year={2003},
  publisher={ACM New York, NY, USA}
}

@inproceedings{blum1993cryptographic,
  title={Cryptographic primitives based on hard learning problems},
  author={Blum, Avrim and Furst, Merrick and Kearns, Michael and Lipton, Richard J},
  booktitle={Annual international cryptology conference},
  pages={278--291},
  year={1993},
  organization={Springer}
}

@inproceedings{lyubashevsky2005parity,
  title={The parity problem in the presence of noise, decoding random linear codes, and the subset sum problem},
  author={Lyubashevsky, Vadim},
  booktitle={International Workshop on Approximation Algorithms for Combinatorial Optimization},
  pages={378--389},
  year={2005},
  organization={Springer}
}

@article{regev2009lattices,
  title={On lattices, learning with errors, random linear codes, and cryptography},
  author={Regev, Oded},
  journal={Journal of the ACM (JACM)},
  volume={56},
  number={6},
  pages={1--40},
  year={2009},
  publisher={ACM New York, NY, USA}
}

@inproceedings{becker2012decoding,
  title={Decoding random binary linear codes in 2 n/20: How 1+ 1= 0 improves information set decoding},
  author={Becker, Anja and Joux, Antoine and May, Alexander and Meurer, Alexander},
  booktitle={Annual international conference on the theory and applications of cryptographic techniques},
  pages={520--536},
  year={2012},
  organization={Springer}
}

@article{raz2018fast,
  title={Fast learning requires good memory: A time-space lower bound for parity learning},
  author={Raz, Ran},
  journal={Journal of the ACM (JACM)},
  volume={66},
  number={1},
  pages={1--18},
  year={2018},
  publisher={ACM New York, NY, USA}
}

@inproceedings{bshouty2024approximating,
  title={Approximating the Number of Relevant Variables in a Parity Implies Proper Learning},
  author={Bshouty, Nader H and Haddad, George},
  booktitle={Approximation, Randomization, and Combinatorial Optimization. Algorithms and Techniques (APPROX/RANDOM 2024)},
  pages={38--1},
  year={2024},
  organization={Schloss Dagstuhl--Leibniz-Zentrum f{\"u}r Informatik}
}

@article{valiant2015finding,
  title={Finding correlations in subquadratic time, with applications to learning parities and the closest pair problem},
  author={Valiant, Gregory},
  journal={Journal of the ACM (JACM)},
  volume={62},
  number={2},
  pages={1--45},
  year={2015},
  publisher={ACM New York, NY, USA}
}

@inproceedings{golowich2024exploration,
  title={Exploration is harder than prediction: Cryptographically separating reinforcement learning from supervised learning},
  author={Golowich, Noah and Moitra, Ankur and Rohatgi, Dhruv},
  booktitle={2024 IEEE 65th Annual Symposium on Foundations of Computer Science (FOCS)},
  pages={1953--1967},
  year={2024},
  organization={IEEE}
}

@inproceedings{arora2011new,
  title={New algorithms for learning in presence of errors},
  author={Arora, Sanjeev and Ge, Rong},
  booktitle={International Colloquium on Automata, Languages, and Programming},
  pages={403--415},
  year={2011},
  organization={Springer}
}

@article{li2025hardness,
  title={Hardness and Algorithms for Batch LPN under Dependent Noise},
  author={Li, Xin and Mao, Songtao and Zhou, Zhaienhe},
  journal={Cryptology ePrint Archive},
  year={2025}
}

@article{dhawan2025detection,
  title={Detection of Dense Subhypergraphs by Low-Degree Polynomials},
  author={Dhawan, Abhishek and Mao, Cheng and Wein, Alexander S},
  journal={Random Structures \& Algorithms},
  volume={66},
  number={1},
  pages={e21279},
  year={2025},
  publisher={Wiley Online Library}
}

@article{mao2021optimal,
  title={Optimal spectral recovery of a planted vector in a subspace},
  author={Mao, Cheng and Wein, Alexander S},
  journal={arXiv preprint arXiv:2105.15081},
  year={2021}
}

@article{ding2024subexponential,
  title={Subexponential-time algorithms for sparse PCA},
  author={Ding, Yunzi and Kunisky, Dmitriy and Wein, Alexander S and Bandeira, Afonso S},
  journal={Foundations of Computational Mathematics},
  volume={24},
  number={3},
  pages={865--914},
  year={2024},
  publisher={Springer}
}

@article{bandeira2022franz,
  title={The Franz-Parisi criterion and computational trade-offs in high dimensional statistics},
  author={Bandeira, Afonso S and El Alaoui, Ahmed and Hopkins, Samuel and Schramm, Tselil and Wein, Alexander S and Zadik, Ilias},
  journal={Advances in Neural Information Processing Systems},
  volume={35},
  pages={33831--33844},
  year={2022}
}

@article{ding2025low,
  title={Low-degree hardness of detection for correlated Erd{\H{o}}s--R{\'e}nyi graphs},
  author={Ding, Jian and Du, Hang and Li, Zhangsong},
  journal={The Annals of Statistics},
  volume={53},
  number={5},
  pages={1833--1856},
  year={2025},
  publisher={Institute of Mathematical Statistics}
}

@article{chen2026computational,
  title={A computational transition for detecting correlated stochastic block models by low-degree polynomials},
  author={Chen, Guanyi and Ding, Jian and Gong, Shuyang and Li, Zhangsong},
  journal={The Annals of Statistics},
  volume={54},
  number={1},
  pages={226--251},
  year={2026},
  publisher={Institute of Mathematical Statistics}
}

@article{gamarnik2024hardness,
  title={Hardness of random optimization problems for Boolean circuits, low-degree polynomials, and Langevin dynamics},
  author={Gamarnik, David and Jagannath, Aukosh and Wein, Alexander S},
  journal={SIAM Journal on Computing},
  volume={53},
  number={1},
  pages={1--46},
  year={2024},
  publisher={SIAM}
}

@article{wein2022optimal,
  title={Optimal low-degree hardness of maximum independent set},
  author={Wein, Alexander S},
  journal={Mathematical Statistics and Learning},
  volume={4},
  number={3},
  pages={221--251},
  year={2022}
}

@article{wang2016average,
  title={Average-case hardness of RIP certification},
  author={Wang, Tengyao and Berthet, Quentin and Plan, Yaniv},
  journal={Advances in Neural Information Processing Systems},
  volume={29},
  year={2016}
}

@inproceedings{bandeira2019computational,
  title={Computational Hardness of Certifying Bounds on Constrained PCA Problems},
  author={Bandeira, Afonso S and Kunisky, Dmitriy and Wein, Alexander S},
  booktitle={11th Innovations in Theoretical Computer Science Conference (ITCS 2020)},
  volume={151},
  year={2020}
}

@inproceedings{kunisky2024computational,
  title={Computational hardness of detecting graph lifts and certifying lift-monotone properties of random regular graphs},
  author={Kunisky, Dmitriy and Yu, Xifan},
  booktitle={2024 IEEE 65th Annual Symposium on Foundations of Computer Science (FOCS)},
  pages={1621--1633},
  year={2024},
  organization={IEEE}
}

@inproceedings{bresler2022algorithmic,
  title={The algorithmic phase transition of random k-sat for low degree polynomials},
  author={Bresler, Guy and Huang, Brice},
  booktitle={2021 IEEE 62nd annual symposium on foundations of computer science (FOCS)},
  pages={298--309},
  year={2022},
  organization={IEEE}
}

@article{li2025algorithmic,
  title={The Algorithmic Phase Transition in Correlated Spiked Models},
  author={Li, Zhangsong},
  journal={arXiv preprint arXiv:2511.06040},
  year={2025}
}

@inproceedings{bangachev2025near,
  title={Near-optimal time-sparsity trade-offs for solving noisy linear equations},
  author={Bangachev, Kiril and Bresler, Guy and Tiegel, Stefan and Vaikuntanathan, Vinod},
  booktitle={Proceedings of the 57th Annual ACM Symposium on Theory of Computing},
  pages={1910--1920},
  year={2025}
}

@inproceedings{arora1995polynomial,
  title={Polynomial time approximation schemes for dense instances of NP-hard problems},
  author={Arora, Sanjeev and Karger, David and Karpinski, Marek},
  booktitle={Proceedings of the twenty-seventh annual ACM symposium on Theory of computing},
  pages={284--293},
  year={1995}
}

@article{haastad2001some,
  title={Some optimal inapproximability results},
  author={H{\aa}stad, Johan},
  journal={Journal of the ACM (JACM)},
  volume={48},
  number={4},
  pages={798--859},
  year={2001},
  publisher={ACM New York, NY, USA}
}

@inproceedings{fotakis2015sub,
  title={Sub-exponential Approximation Schemes for CSPs: From Dense to Almost Sparse},
  author={Fotakis, Dimitris and Lampis, Michail and Paschos, Vangelis},
  booktitle={33rd Symposium on Theoretical Aspects of Computer Science (STACS 2016)},
  pages={37--1},
  year={2016}
}

@article{impagliazzo2001complexity,
  title={On the complexity of k-SAT},
  author={Impagliazzo, Russell and Paturi, Ramamohan},
  journal={Journal of Computer and System Sciences},
  volume={62},
  number={2},
  pages={367--375},
  year={2001},
  publisher={Elsevier}
}

@article{barthel2002hiding,
  title={Hiding solutions in random satisfiability problems: A statistical mechanics approach},
  author={Barthel, Wolfgang and Hartmann, Alexander K and Leone, Michele and Ricci-Tersenghi, Federico and Weigt, Martin and Zecchina, Riccardo},
  journal={Physical review letters},
  volume={88},
  number={18},
  pages={188701},
  year={2002},
  publisher={APS}
}

@article{jia2007generating,
  title={Generating hard satisfiable formulas by hiding solutions deceptively},
  author={Jia, Haixia and Moore, Cristopher and Strain, Doug},
  journal={Journal of Artificial Intelligence Research},
  volume={28},
  pages={107--118},
  year={2007}
}

@inproceedings{bogdanov2009security,
  title={On the security of goldreich’s one-way function},
  author={Bogdanov, Andrej and Qiao, Youming},
  booktitle={International Workshop on Approximation Algorithms for Combinatorial Optimization},
  pages={392--405},
  year={2009},
  organization={Springer}
}

@article{coja2010efficient,
  title={An efficient sparse regularity concept},
  author={Coja-Oghlan, Amin and Cooper, Colin and Frieze, Alan},
  journal={SIAM Journal on Discrete Mathematics},
  volume={23},
  number={4},
  pages={2000--2034},
  year={2010},
  publisher={SIAM}
}

@inproceedings{abascal2021strongly,
  title={Strongly refuting all semi-random Boolean CSPs},
  author={Abascal, Jackson and Guruswami, Venkatesan and Kothari, Pravesh K},
  booktitle={Proceedings of the 2021 ACM-SIAM Symposium on Discrete Algorithms (SODA)},
  pages={454--472},
  year={2021},
  organization={SIAM}
}

@inproceedings{guruswami2023efficient,
  title={Efficient algorithms for semirandom planted csps at the refutation threshold},
  author={Guruswami, Venkatesan and Hsieh, Jun-Ting and Kothari, Pravesh K and Manohar, Peter},
  booktitle={2023 IEEE 64th Annual Symposium on Foundations of Computer Science (FOCS)},
  pages={307--327},
  year={2023},
  organization={IEEE}
}

@article{alon2008biomolecular,
  title={Biomolecular network motif counting and discovery by color coding},
  author={Alon, Noga and Dao, Phuong and Hajirasouliha, Iman and Hormozdiari, Fereydoun and Sahinalp, S Cenk},
  journal={Bioinformatics},
  volume={24},
  number={13},
  pages={i241--i249},
  year={2008},
  publisher={Oxford University Press}
}

@article{buhai2025quasi,
  title={The quasi-polynomial low-degree conjecture is false},
  author={Buhai, Rares-Darius and Hsieh, Jun-Ting and Jain, Aayush and Kothari, Pravesh K},
  journal={arXiv preprint arXiv:2505.17360},
  year={2025}
}

@inproceedings{holmgren2020counterexamples,
  title={Counterexamples to the Low-Degree Conjecture},
  author={Holmgren, Justin and Wein, Alexander S},
  booktitle={12th Innovations in Theoretical Computer Science Conference (ITCS 2021)},
  pages={75--1},
  year={2021},
  organization={Schloss Dagstuhl--Leibniz-Zentrum f{\"u}r Informatik}
}

@inproceedings{bogdanov2025sample,
  title={Sample Efficient Search to Decision for k LIN},
  author={Bogdanov, Andrej and Rosen, Alon and Tan, Kel Zin},
  booktitle={Annual International Cryptology Conference},
  pages={203--220},
  year={2025},
  organization={Springer}
}

@article{mao2024testing,
  title={Testing network correlation efficiently via counting trees},
  author={Mao, Cheng and Wu, Yihong and Xu, Jiaming and Yu, Sophie H},
  journal={The Annals of Statistics},
  volume={52},
  number={6},
  pages={2483--2505},
  year={2024},
  publisher={Institute of Mathematical Statistics}
}

@inproceedings{chen2026detecting,
  title={Detecting Correlation Efficiently in Very Supercritical Stochastic Block Models: Breaking the Otter’s Threshold Barrier},
  author={Chen, Guanyi and Ding, Jian and Gong, Shuyang and Li, Zhangsong},
  booktitle={Proceedings of the 2026 Annual ACM-SIAM Symposium on Discrete Algorithms (SODA)},
  pages={2743--2759},
  year={2026},
  organization={SIAM}
}
\appendix
\section{Asymptotic Estimates for Combinatorial Quantities}

In this appendix, we record several elementary asymptotic estimates that are used throughout the paper. We state them here for convenience and omit further justification in the main text.

\begin{lemma}[Stirling-type bounds]\label{lem:stirling}
For every integer $n\ge 1$,
\[
\sqrt{n}\left(\frac{n}{e}\right)^n
\;\le\;
n!
\;\le\;
e\sqrt{n}\left(\frac{n}{e}\right)^n.
\]
\end{lemma}

\begin{proof}
This is an immediate consequence of Stirling's formula.
\end{proof}

\begin{lemma}[Falling factorial asymptotics]\label{lem:falling-factorial-asymp}
If $v=o(\sqrt n)$, then
\[
(n)_v
=
n^v\left(1+O\!\left(\frac{v^2}{n}\right)\right)
=
(1+o(1))\,n^v.
\]
In particular,
\[
(1-o(1))\,n^v \le (n)_v \le n^v.
\]
\end{lemma}

\begin{proof}
We write
\[
\frac{(n)_v}{n^v}
=
\prod_{t=0}^{v-1}\left(1-\frac{t}{n}\right).
\]
Since $v=o(\sqrt n)$, we have $t/n=o(1)$ uniformly for $0\le t\le v-1$. Thus
\[
\log\frac{(n)_v}{n^v}
=
\sum_{t=0}^{v-1}\log\left(1-\frac{t}{n}\right)
=
-\sum_{t=0}^{v-1}\frac{t}{n}
+
O\!\left(\sum_{t=0}^{v-1}\frac{t^2}{n^2}\right).
\]
Now
\[
\sum_{t=0}^{v-1}\frac{t}{n}
=
\frac{v(v-1)}{2n}
=
O\!\left(\frac{v^2}{n}\right),
\]
and
\[
\sum_{t=0}^{v-1}\frac{t^2}{n^2}
=
O\!\left(\frac{v^3}{n^2}\right)
=
o\!\left(\frac{v^2}{n}\right),
\]
since $v=o(\sqrt n)$. Hence
\[
\log\frac{(n)_v}{n^v}
=
O\!\left(\frac{v^2}{n}\right).
\]
Exponentiating gives
\[
\frac{(n)_v}{n^v}
=
1+O\!\left(\frac{v^2}{n}\right),
\]
as claimed. The upper bound $(n)_v\le n^v$ is immediate.
\end{proof}

\begin{lemma}[Binomial coefficient asymptotics]\label{lem:binom-asymp-small-v}
If $v=o(\sqrt n)$, then
\[
\binom{n}{v}
=
\frac{n^v}{v!}\left(1+O\!\left(\frac{v^2}{n}\right)\right)
=
(1+o(1))\frac{n^v}{v!}.
\]
\end{lemma}

\begin{proof}
This follows immediately from
\[
\binom{n}{v}=\frac{(n)_v}{v!}
\]
and Lemma~\ref{lem:falling-factorial-asymp}.
\end{proof}

\begin{lemma}[Ratio of nearby falling factorials]\label{lem:falling-factorial-ratio}
Let $v=o(\sqrt n)$ and $0\le i\le v$. Then uniformly over all $0\le i\le v$,
\[
\frac{(n)_{v-i}}{(n)_v}
=
\frac{1}{(n-v+i)_i}
=
n^{-i}\left(1+O\!\left(\frac{vi}{n}\right)\right)
=
(1+o(1))\,n^{-i}.
\]
\end{lemma}

\begin{proof}
We have
\[
\frac{(n)_{v-i}}{(n)_v}
=
\frac{1}{(n-v+i)_i}.
\]
Also,
\[
(n-v+i)_i
=
n^i\prod_{t=0}^{i-1}\left(1-\frac{v-i-t}{n}\right).
\]
Since $0\le i\le v=o(\sqrt n)$, the quantity $(v-i+t)/n=o(1)$ uniformly. Hence
\[
\log\frac{(n-v+i)_i}{n^i}
=
\sum_{t=0}^{i-1}\log\left(1-\frac{v-i-t}{n}\right)
=
O\!\left(\frac{vi}{n}\right).
\]
Exponentiating yields
\[
(n-v+i)_i=n^i\left(1+O\!\left(\frac{vi}{n}\right)\right),
\]
and therefore
\[
\frac{(n)_{v-i}}{(n)_v}
=
n^{-i}\left(1+O\!\left(\frac{vi}{n}\right)\right).
\]
Since $vi\le v^2=o(n)$, this is also $(1+o(1))n^{-i}$.
\end{proof}

\begin{lemma}\label{lem:hypergeom-asymp}
Let $v=o(\sqrt n)$. Then, uniformly for all $0\le i\le v$,
\[
\frac{\binom{v}{i}\binom{n-v}{v-i}}{\binom{n}{v}}
=
\frac{(v)_i^2}{i!\,n^i}\left(1+O\!\left(\frac{v^2}{n}\right)\right)
=
(1+o(1))\frac{(v)_i^2}{i!\,n^i}.
\]
\end{lemma}

\begin{proof}
We compute
\[
\frac{\binom{v}{i}\binom{n-v}{v-i}}{\binom{n}{v}}
=
\frac{(v)_i}{i!}\cdot
\frac{(n-v)_{v-i}}{(v-i)!}\cdot
\frac{v!}{(n)_v}
=
\frac{(v)_i^2}{i!\,(n)_i}\cdot
\frac{(n-v)_{v-i}}{(n-i)_{v-i}}.
\]
For the second factor,
\[
\frac{(n-v)_{v-i}}{(n-i)_{v-i}}
=
\prod_{t=0}^{v-i-1}
\left(1-\frac{v-i}{n-i-t}\right).
\]
Since $v=o(\sqrt n)$, we have $v=o(n)$, and uniformly for all $0\le t\le v-i-1$,
\[
0\le \frac{v-i}{n-i-t}\le \frac{v}{n-v+1}=o(1).
\]
Therefore,
\[
\log\left(\frac{(n-v)_{v-i}}{(n-i)_{v-i}}\right)
=
\sum_{t=0}^{v-i-1}
\log\left(1-\frac{v-i}{n-i-t}\right)
\]
\[
=
-\sum_{t=0}^{v-i-1}\frac{v-i}{n-i-t}
+
O\!\left(
\sum_{t=0}^{v-i-1}\frac{(v-i)^2}{(n-i-t)^2}
\right).
\]
Now
\[
\sum_{t=0}^{v-i-1}\frac{v-i}{n-i-t}
\le
\frac{(v-i)^2}{n-v+1}
=
O\!\left(\frac{v^2}{n}\right),
\]
and
\[
\sum_{t=0}^{v-i-1}\frac{(v-i)^2}{(n-i-t)^2}
\le
\frac{(v-i)^3}{(n-v+1)^2}
=
O\!\left(\frac{v^3}{n^2}\right)
=
o\!\left(\frac{v^2}{n}\right).
\]
Hence
\[
\log\left(\frac{(n-v)_{v-i}}{(n-i)_{v-i}}\right)
=
O\!\left(\frac{v^2}{n}\right),
\]
and so
\[
\frac{(n-v)_{v-i}}{(n-i)_{v-i}}
=
1+O\!\left(\frac{v^2}{n}\right).
\]
Also, by Lemma~\ref{lem:falling-factorial-asymp},
\[
(n)_i=n^i\left(1+O\!\left(\frac{v^2}{n}\right)\right)
\]
uniformly for $0\le i\le v$. Substituting these estimates into the previous identity yields
\[
\frac{\binom{v}{i}\binom{n-v}{v-i}}{\binom{n}{v}}
=
\frac{(v)_i^2}{i!\,n^i}\left(1+O\!\left(\frac{v^2}{n}\right)\right),
\]
as claimed.
\end{proof}

\begin{lemma}[Elementary product bound]\label{lem:prod-exp-bound}
For any real numbers $x_1,\dots,x_m\in[0,1)$,
\[
\prod_{t=1}^m(1-x_t)
\le
\exp\!\left(-\sum_{t=1}^m x_t\right).
\]
\end{lemma}

\begin{proof}
This follows from the inequality $\log(1-x)\le -x$ for $0\le x<1$.
\end{proof}
\section{Equivalence of Different Sampling Models for Noisy $k$-XOR}\label{appex:models}

The goal of this section is to show that, in the parameter regime relevant to our results, the three sampling models introduced in Section~\ref{subsec:equi} are asymptotically equivalent. We first compare the Bernoulli-$p$ model with the fixed-$m$ model without replacement, and then compare the latter with the standard fixed-$m$ model with replacement. Recall that
\[
\sM:=\binom{[n]}{k},
\qquad
M:=|\sM|=\binom{n}{k}.
\]

\begin{lemma}
Let
\[
N:=|S|=\sum_{\alpha\in\sM} b_\alpha \sim \Bin(M,p),
\]
where $S=\{\alpha\in\sM:b_\alpha=1\}$ is the observed support in the Bernoulli-$p$ model. Then $\E[N]=pM$, and for every $t>0$,
\[
\Pr\bigl(|N-pM|\ge t\bigr)
\le
2\exp\!\left(-\frac{t^2}{2pM+\frac23 t}\right).
\]
In particular, for any fixed $C>0$,
\[
|N-pM|
\le
C\sqrt{pM\log M}
\]
with probability at least $1-2M^{-c_C}$, where $c_C>0$ depends only on $C$.
\end{lemma}

This is the standard Bernstein--Chernoff bound for binomial random variables. Therefore, if we choose
\[
p=\frac{m}{M},
\]
then $N$ is sharply concentrated around $m$. In particular,
\[
N=m+O\!\left(\sqrt{m\log M}\right)
\]
with high probability.

\begin{lemma}
Conditioned on the event $\{N=t\}$, the support $S$ in the Bernoulli-$p$ model is a uniformly random $t$-subset of $\sM$. Moreover, conditional on $S$, the observed labels $(z(\alpha))_{\alpha\in S}$ are independent and follow the same planted/noise law as in the fixed-$t$ model without replacement. In particular, conditioned on $\{N=m\}$, the Bernoulli-$p$ model agrees exactly with the fixed-$m$ model without replacement.
\end{lemma}

\begin{proof}
Since the indicators $(b_\alpha)_{\alpha\in\sM}$ are i.i.d. Bernoulli random variables, conditioning on $\sum_{\alpha\in\sM} b_\alpha=t$ makes the support $S$ uniformly distributed over all $t$-subsets of $\sM$. The labels remain independent conditional on the support because the noise variables $(\xi_\alpha)$ are independent of the indicators $(b_\alpha)$.
\end{proof}

We next compare the fixed-$m$ model with replacement to the fixed-$m$ model without replacement.

\begin{lemma}
Let $E_{\mathrm{coll}}$ denote the event that a repeated constraint is sampled in Model~1, namely
\[
\alpha_s=\alpha_t
\qquad
\text{for some } s\neq t.
\]
Then
\[
\Pr(E_{\mathrm{coll}})
\le
\binom{m}{2}\frac{1}{M}
\le
\frac{m(m-1)}{2M}.
\]
Consequently, if $m=o(M^{1/2})$, then
\[
\Pr(E_{\mathrm{coll}})=o(1).
\]
\end{lemma}

\begin{lemma}
Conditioned on $E_{\mathrm{coll}}^c$, the unordered support
\[
\{\alpha_1,\dots,\alpha_m\}
\]
in the fixed-$m$ model with replacement is a uniformly random $m$-subset of $\sM$. Moreover, conditional on this support, the observed labels are independent and follow the same planted/noise law as in the fixed-$m$ model without replacement. Hence, conditioned on $E_{\mathrm{coll}}^c$, the fixed-$m$ model with replacement agrees exactly with the fixed-$m$ model without replacement.
\end{lemma}

\begin{proof}
Conditional on all sampled constraints being distinct, the ordered tuple $(\alpha_1,\dots,\alpha_m)$ is uniformly distributed over all ordered $m$-tuples of distinct elements of $\sM$. Therefore the induced unordered support is uniformly distributed over all $m$-subsets of $\sM$. The labels remain independent because the noise variables are independent across samples.
\end{proof}

We can now state the main equivalence result.

\begin{proposition}[Equivalence of the sampling models]\label{prop:model-equivalence}
Set
\[
p=\frac{m}{M}=\frac{m}{\binom{n}{k}}.
\]
Then the fixed-$m$ model with replacement and the Bernoulli-$p$ model can be coupled so that, under either the planted or null law, the following holds with probability at least
\[
1-\Pr(E_{\mathrm{coll}})-\Pr(|N-m|>r):
\]
after passing through the fixed-$m$ model without replacement, the two observations differ by at most $r$ constraint insertions/deletions.
\end{proposition}

\begin{proof}
Introduce the fixed-$m$ model without replacement, denoted by $\PP^{(2)}_{m,\delta}$ in the planted case and $\QQ^{(2)}_m$ in the null case.

By Lemma~..., the fixed-$m$ model with replacement agrees exactly with the fixed-$m$ without-replacement model on the event $E_{\mathrm{coll}}^c$.

Next consider the Bernoulli-$p$ model. On the event $|N-m|\le r$, we convert it to a model with exactly $m$ observed constraints as follows. If $N\ge m$, discard $N-m$ observed constraints uniformly at random. If $N<m$, append $m-N$ additional constraints chosen uniformly from the remaining set $\sM\setminus S$, together with independent labels drawn from the same planted (respectively null) law. The resulting observation has exactly $m$ constraints and agrees in law with the fixed-$m$ model without replacement.

Combining the two reductions yields the claimed coupling statement.
\end{proof}

Suppose $m=o(M^{1/2})$, equivalently $m=o(n^{k/2})$. Then the collision probability in the fixed-$m$ model with replacement is $o(1)$. Also, if $p=m/M$, then the number $N$ of observed constraints in the Bernoulli-$p$ model is concentrated around $m$, with fluctuations of order $O(\sqrt{m\log M})$. Therefore, in this regime, the discrepancy between the three sampling models is asymptotically negligible for the purposes of distinguishing and recovery.

More precisely, if an algorithm succeeds with probability $1-\varepsilon$ under one model, then under either of the other two models its success probability is at least
\[
1-\varepsilon-\Pr(E_{\mathrm{coll}})-\Pr(|N-m|>r),
\]
after applying the above truncation/padding reduction with parameter $r$.
Suppose $m=o(M^{1/2})$, equivalently $m=o(n^{k/2})$. Then the collision probability in the fixed-$m$ model with replacement is $o(1)$. Also, if $p=m/M$, then the number $N$ of observed constraints in the Bernoulli-$p$ model is concentrated around $m$, with fluctuations of order $O(\sqrt{m\log M})$. Therefore the two models differ only by an asymptotically negligible perturbation, and any distinguishing or recovery procedure for one model immediately yields a corresponding procedure for the other with essentially the same success probability.

More precisely, if an algorithm succeeds with probability $1-\varepsilon$ under one model, then under the other model its success probability is at least
\[
1-\varepsilon-\Pr(E_{\mathrm{coll}})-\Pr(|N-m|>r),
\]
for a suitable truncation parameter $r$.
\section{Postponed Proofs}

\subsection{Proof of Lemma \ref{lem:Meanvariance}}\label{appex:mean-var}
\begin{proof}
By Lemma~\ref{lem:mean-variance-separation}, together with Equation~\eqref{eq:VarP-i-j} and Lemma~\ref{lem:Nij-zero}, it follows that $\mathcal A$ distinguishes $\PP_x$ from $\QQ$ with success probability at least
$$
1
\;-\;
O\!\left(
\frac{2v!}{ |\mathcal H|\,(n)_v\,p^s \delta^{2s}}
\;+\;
\sum_{i=k}^{v-1}\sum_{j=1}^{\lfloor(2i-2)/k\rfloor}
\frac{N(i,j)\,p^{2s-j}\delta^{2s-2j}}{|\mathcal H|^2\,(n)_v^2\,p^{2s} \delta^{2s}}
\right).
$$
This follows directly from the above mean and variance computations together with
Lemma~\ref{lem:mean-variance-separation}.
Since $p = \frac{m}{\binom{n}{k}}$.
$$
|\mathcal H|\geq \frac{1}{(2\ell)^2}\cdot
\left(C\cdot
\frac{(2kr)!}{2^{rk}\,(k!)^{r} (2r)!}
\right)^{\ell} $$

\[
\frac{v!}{|\mathcal H|\,(n)_v\,p^s\,\delta^{2s}}
\le
\frac{(2\ell)^2}{C^\ell\,\delta^{2s}}
\cdot
\frac{v!}{(n)_v}
\cdot
\left(\frac{2^{rk}(k!)^r(2r)!}{(2kr)!}\right)^\ell
\cdot
\left(\frac{\binom{n}{k}}{m}\right)^s.
\]
Substituting
\[
v=rk\ell,\qquad s=2r\ell,
\]
we obtain
\[
\frac{v!}{|\mathcal H|\,(n)_v\,p^s\,\delta^{2s}}
\le
\frac{(2\ell)^2}{C^\ell\,\delta^{4r\ell}}
\cdot
\frac{(rk\ell)!}{(n)_{rk\ell}}
\cdot
\left(
\left(\frac{2^{rk}(k!)^r(2r)!}{(2kr)!}\right)^{1/(2r)}
\cdot
\frac{\binom{n}{k}}{m}
\right)^{2r\ell}.
\]
Since $rk\ell\ll \sqrt n$, we have
\[
(n)_{rk\ell}=(1-o(1))\,n^{rk\ell}.
\]
Also,
\[
(rk\ell)!
< e \sqrt{ rk\ell}\left(\frac{rk\ell}{e}\right)^{rk\ell}.
\]
Therefore
\[
\frac{(rk\ell)!}{(n)_{rk\ell}}
\leq 
e
\sqrt{rk\ell}
\left(\frac{rk\ell}{en}\right)^{rk\ell}.
\]
Using
\[
k!\le e\sqrt{k}\left(\frac{k}{e}\right)^k,\qquad
(2r)!\le e\sqrt{2r}\left(\frac{2r}{e}\right)^{2r},\qquad
(2kr)!\ge \sqrt{2kr}\left(\frac{2kr}{e}\right)^{2kr},
\]
we obtain
\begin{align*}
\frac{2^{rk}(k!)^r(2r)!}{(2kr)!}
&\le
\frac{
2^{rk}\bigl(e\sqrt{k}(k/e)^k\bigr)^r\cdot e\sqrt{2r}(2r/e)^{2r}
}{
\sqrt{2kr}(2kr/e)^{2kr}
}
\\
&=
\frac{e^{r+1}k^{r/2}\,k^{kr}e^{-kr}\,(2r)^{2r}e^{-2r}\,2^{rk}}
{\sqrt{k}\,(2kr)^{2kr}e^{-2kr}}
\\
&=
\frac{e^{r(k-1)+1}}{\sqrt{k}}\,
\frac{2^{rk}(2r)^{2r}}{2^{2kr}k^{kr}r^{2kr}}\,
k^{r/2}
\\
&=
\frac{e^{r(k-1)+1}}{\sqrt{k}}\,
2^{\,2r-rk}\,
k^{\frac r2-kr}\,
r^{\,2r-2kr}
\\
&\le
\left(
e^{k}2^{\,2-k}k^{-(k-\frac12)}r^{-2(k-1)}
\right)^r.
\end{align*}
together with
\[
\binom{n}{k}\le \frac{n^k}{k!},
\]
we obtain
\[
\begin{aligned}
\frac{v!}{|\mathcal H|\,(n)_v\,p^s\,\delta^{2s}}
&\le
\frac{(2\ell)^2}{C^\ell\,\delta^{4r\ell}}
\cdot
e\sqrt{rk\ell}
\left(\frac{rk\ell}{en}\right)^{rk\ell}
\cdot
\left(
\left(
\frac{e^{k}}{2^{k-2}\,r^{2k-2}\,k^{\,k-\frac12}}
\right)^{1/2}
\frac{n^k}{k!\,m}
\right)^{2r\ell}
\\
&=
\frac{(2\ell)^2 e\sqrt{rk\ell}}{C^\ell}
\left[
\left(\frac{rk\ell}{en}\right)^{k/2}
\cdot
\frac{e^{k/2}}{2^{k/2-1}\,r^{k-1}\,k^{\,k/2-\frac14}}
\cdot
\frac{n^k}{k!\,m\,\delta^2}
\right]^{2r\ell}
\\
&=
\frac{(2\ell)^2 e\sqrt{rk\ell}}{C^\ell}
\left[
\frac{r^{k/2}k^{k/2}\ell^{k/2}}{e^{k/2}n^{k/2}}
\cdot
\frac{e^{k/2}}{2^{k/2-1}\,r^{k-1}\,k^{\,k/2-\frac14}}
\cdot
\frac{n^k}{k!\,m\,\delta^2}
\right]^{2r\ell}
\\
&=
\frac{(2\ell)^2 e\sqrt{rk\ell}}{C^\ell}
\left[
\frac{k^{1/4}\,\ell^{k/2}}{2^{k/2-1}\,r^{k/2-1}}
\cdot
\frac{n^{k/2}}{k!\,m\,\delta^2}
\right]^{2r\ell}
\\
&=
\frac{(2\ell)^2 e\sqrt{rk\ell}}{C^\ell}
\left[
\frac{k^{1/4}\,\ell^{k/2}}{k!\,2^{k/2-1}}
\cdot
\frac{n^{k/2}}{m\,\delta^2\,r^{k/2-1}}
\right]^{2r\ell}
\\
&\le
\left[
\frac{4\,\ell^{k/2} k^{k/2}}{k!\,2^{k/2}}
\cdot
\frac{n^{k/2}}{m\,\delta^2\,(kr)^{k/2-1}}
\right]^{2r\ell},
\end{aligned}
\]
when $r$ is larger than a some constant.

We next turn to the second term. By Lemma~\ref{lem:Nij-sum},
\[
\sum_{i=k}^{v-1}\sum_{j=1}^{\lfloor(2i-2)/k\rfloor}
\frac{N(i,j)\,p^{2s-j}\delta^{2s-2j}}{|\mathcal H|^2\,(n)_v^2\,p^{2s} \delta^{2s}}\leq
\sum_{i=k}^{v-1}
n^{-i}\binom{v}{i}\,
p^{-\lfloor(2i-2)/k\rfloor}
\delta^{-2\lfloor(2i-2)/k\rfloor}
\]
Let
\[
a_i:=\left\lfloor\frac{2i-2}{k}\right\rfloor,
\qquad
Y:=\frac{1}{n^{k/2}p\delta^2}.
\]
Then
\[
\sum_{i=k}^{v-1}\sum_{j=1}^{\lfloor(2i-2)/k\rfloor}
\frac{N(i,j)\,p^{2s-j}\delta^{2s-2j}}{|\mathcal H|^2\,(n)_v^2\,p^{2s} \delta^{2s}}
\le
\frac{1}{n}\sum_{i=k}^{v-1}\binom{v}{i}Y^{a_i}
\]
As above,
\[
\sum_{i=k}^{v-1}\binom{v}{i}Y^{a_i}
\le
k\,v^k
\sum_{j=0}^{\lfloor 2(v-1)/k\rfloor}
\bigl(v^{k/2}Y\bigr)^j.
\]
If $v^{k/2}Y>1$, then the summand grows geometrically in $i$, so the sum is dominated by its last term. Hence
\[
\begin{aligned}
\sum_{i=k}^{v-1}\binom{v}{i}Y^{a_i}
&\le
k\,v^k
\left(\left\lfloor \frac{2(v-1)}{k}\right\rfloor+1\right)
\bigl(v^{k/2}Y\bigr)^{\lfloor 2(v-1)/k\rfloor}\\
&\le
2v^{k+1}\bigl(v^{k/2}Y\bigr)^{2v/k}\\
&=
2v^{k+1}v^v
\left(\frac{\binom{n}{k}}{n^{k/2}m\delta^2}\right)^{2v/k}\\
&\le
2v^{k+1}v^v v^{2v/k}
\left(\frac{n^{k/2}}{(rk)^{k/2-1}\, k!\,m\delta^2}\right)^{2r\ell}.
\end{aligned}
\]
Now assume
\[
v\le \frac{0.3\log n}{\log\log n}.
\]
Then each of the factors $v^{k+1}$, $v^v$, and $v^{2v/k}$ is at most $n^{0.3}$ for all sufficiently large $n$, and therefore
\[
2v^{k+1}v^v v^{2v/k}\le n^{0.9}.
\]
It follows that
\[
\frac1n\sum_{i=k}^{v-1}\binom{v}{i}Y^{a_i}
\le
\frac{1}{n^{0.1}}
\left(\frac{n^{k/2}}{(rk)^{k/2-1}\,k! m\delta^2}\right)^{2r\ell}.
\]
\end{proof}
\subsection{Proof of Lemma \ref{lem:Nij-zero-recovery}}\label{appex:proof-Nij-recovery}
\begin{proof}
Fix a tuple $(J,J',\phi,\psi)$ counted by $N_{a,b}(i,j)$, and write
\[
W:=\phi(V(J))\cap \psi(V(J')),\qquad |W|=i,
\]
and
\[
K:=E(\phi(J))\cap E(\psi(J')), \qquad |K|=j.
\]
Thus every edge in $K$ is a $k$-subset of $W$, and hence the total number of
edge--vertex incidences in the overlap hypergraph $(W,K)$ is exactly
\[
\sum_{u\in W}\deg_{K}(u)=jk.
\]
We show that, unless $(i,j)=(v,s)$, one has
\[
jk\le 2i-6.
\]
This immediately implies that $N_{a,b}(i,j)=0$ whenever $2i\le jk+5$.

We divide into two cases.

\medskip
\noindent
\textbf{Case 1: some interface vertex $v_t$ is not contained in $W$.}

Since $\phi(v_1)=\psi(v_1)=a$ and $\phi(v_{\ell+1})=\psi(v_{\ell+1})=b$, the overlap
set $W$ contains both endpoints $a,b$. If some interface vertex $v_t$ is missing from
$W$, then the overlap hypergraph $(W,K)$ cannot connect $a$ to $b$; indeed, any path in
$J$ from $v_1$ to $v_{\ell+1}$ must pass through the successive interface vertices
$v_2,\dots,v_\ell$. Hence $(W,K)$ has at least two connected components, one containing
$a$ and one containing $b$.

Let $W_1,\dots,W_m$ be the connected components of $(W,K)$, where $m\ge 2$, and let
$i_r:=|W_r|$. For each component $W_r$, because it is a proper overlap component, at least
one vertex of $W_r$ is incident in $\phi(J)$ to an edge leaving $W_r$, and at least one
(possibly different) vertex of $W_r$ is incident in $\psi(J')$ to an edge leaving $W_r$.
These two outgoing edges are distinct, since one belongs to $\phi(J)\setminus K$ and the other
belongs to $\psi(J')\setminus K$.

Now every vertex has degree at most $2$ in each of $\phi(J)$ and $\psi(J')$, while the
endpoints $a,b$ have degree $1$. Therefore, if $W_r$ contains neither $a$ nor $b$, then
the total degree available inside $W_r$ is at most $2i_r$, and the two outgoing edges force
\[
\sum_{u\in W_r}\deg_K(u)\le 2i_r-2.
\]
If $W_r$ contains exactly one of $a,b$, then its total degree budget is at most $2i_r-1$,
and again two outgoing edges are lost, so
\[
\sum_{u\in W_r}\deg_K(u)\le 2i_r-3.
\]
Since there are at least two components containing $a$ and $b$, summing over all components
gives
\[
jk=\sum_{r=1}^m \sum_{u\in W_r}\deg_K(u)
\le
\sum_{r=1}^m (2i_r-2)-2
=
2i-2m-2
\le
2i-6.
\]

\medskip
\noindent
\textbf{Case 2: all interface vertices $v_1,\dots,v_{\ell+1}$ are contained in $W$.}

Since we are excluding the diagonal case $i=v$, not all vertices are shared. Hence there exists
some block $U_t$ such that
\[
V(U_t)\setminus \{v_t,v_{t+1}\}\not\subseteq W.
\]
In other words, the overlap inside $U_t$ is a proper vertex-subhypergraph containing the two
leaves $v_t,v_{t+1}$, but missing at least one internal vertex.

Now use property~(3) in Definition~\ref{def:ugraph}: deleting any single vertex from $U_t$
leaves a connected hypergraph. This implies that in $U_t$, no proper subhypergraph containing
both leaves can be separated from the rest by removing only one edge. In particular, the overlap
inside $U_t$ must have at least two edges of $\phi(J)$ leaving it, and likewise at least two
edges of $\psi(J')$ leaving it. These four outgoing edges are all distinct, because the two
families of outgoing edges belong respectively to $\phi(J)\setminus K$ and $\psi(J')\setminus K$.

Since all interface vertices lie in $W$, the only vertices of degree $1$ in the overlap are
the two endpoints $a,b$. Therefore the total degree budget of vertices in $W$ inside the full
patterns is at most
\[
2(i-2)+1+1 = 2i-2.
\]
Removing the four outgoing edges just discussed, we obtain
\[
jk=\sum_{u\in W}\deg_K(u)\le 2i-2-4=2i-6.
\]

Combining the two cases, we conclude that every non-diagonal overlap satisfies
\[
jk\le 2i-6.
\]
Therefore, if $2i\le jk+5$, then no such non-diagonal tuple can exist, and hence
\[
N_{a,b}(i,j)=0,
\]
as claimed.

It remains to treat the diagonal case $(i,j)=(v,s)$. In this case the two rooted embeddings
have the same image hypergraph. Thus, after choosing $J\in\mathcal J$ and a rooted embedding
$\phi$ with
\[
\phi(v_1)=a,\qquad \phi(v_{\ell+1})=b,
\]
there are at most $(v-2)!$ choices for $(J',\psi)$, corresponding to arbitrary relabelings
of the $v-2$ internal vertices while keeping the two distinguished endpoints fixed. Since the
number of rooted embeddings $\phi$ is $(n-2)_{v-2}$, we get
\[
N_{a,b}(v,s)=|\mathcal J|\,(n-2)_{v-2}\,(v-2)!.
\]
This completes the proof.
\end{proof}
\subsection{Proof of Lemma \ref{lem:Meanvariance-recovery}}
\label{apx:recover-proof}
\begin{proof}
By Lemma~\ref{lem:Nij-zero-recovery} and Lemma~\ref{lem:rooted-overlap-fixed-i}, we have
\[
\begin{aligned}
\Var(F_{\mathcal J,a,b}(z))
&\le
\sum_{i=2}^{v}\sum_{j=1}^{s}
N_{a,b}(i,j)\,p^{2s-j}\delta^{2s-2j}\\
&\le
\sum_{i=2}^{v}\sum_{j=1}^{\lfloor(2i-6)/k\rfloor}
N_{a,b}(i,j)\,p^{2s-j}\delta^{2s-2j}
+
|\mathcal J|\,(n-2)_{v-2}\,(v-2)!\,p^s .
\end{aligned}
\]
Therefore,
\[
\begin{aligned}
\frac{\Var(F_{\mathcal J,a,b}(z))}
{\E[F_{\mathcal J,a,b}(z)]^2}
&\le
\frac{
\sum_{i=2}^{v}\sum_{j=1}^{\lfloor(2i-6)/k\rfloor}
N_{a,b}(i,j)\,p^{2s-j}\delta^{2s-2j}
+
|\mathcal J|\,(n-2)_{v-2}\,(v-2)!\,p^s
}{
|\mathcal J|^2 (n-2)_{v-2}^2 (p\delta)^{2s}
}\\
&\le
\sum_{i=2}^{v}
\frac{
\binom{v-2}{i-2}(n-v)_{v-i}\,
p^{-\lfloor(2i-6)/k\rfloor}\delta^{-2\lfloor(2i-6)/k\rfloor}
}{
(n-2)_{v-2}
}
+
\frac{(v-2)!}{|\mathcal J|(n-2)_{v-2}\,p^s\delta^{2s}}
\\
&\le
\sum_{i=2}^{v}
\binom{v-2}{i-2}(n-v)^{-i+2}\,
p^{-\lfloor(2i-6)/k\rfloor}\delta^{-2\lfloor(2i-6)/k\rfloor}
+
\frac{(v-2)!}{|\mathcal J|(n-2)_{v-2}\,p^s\delta^{2s}}.
\end{aligned}
\]
Using the lower bound on $|\mathcal J|$ from Lemma~\ref{lem:counting-J}, the remaining estimates are entirely analogous to the $\Var/\E^2$ calculation in Section~\ref{appex:mean-var}. The only differences are that $n$ is replaced by $n-2$, and the overlap parameter $i$ is replaced by $i-2$. Since $v\ll \sqrt n$, these modifications affect the final bound only by a multiplicative $1+o(1)$ factor. Hence, if we define
\[
\delta_n
:=
\left[
\frac{4\,\ell^{k/2} k^{k/2}}{k!\,2^{k/2}}
\cdot
\frac{n^{k/2}}{m\,\delta^2\,(kr)^{k/2-1}}
\right]^{2r\ell}
+
\frac{1}{n^{0.1}}
\left(
\frac{n^{k/2}}{k!\,m\,\delta^2\,(rk)^{k/2-1}}
\right)^{2r\ell},
\]
then for every $a\neq b$,
\[
\frac{\Var(F_{\mathcal J,a,b}(z))}
{\E[F_{\mathcal J,a,b}(z)]^2}
\le \delta_n.
\]

Now fix $a\neq b$. By \eqref{eq:E-FJab},
\[
\E[F_{\mathcal J,a,b}(z)]
=
|\mathcal J|(n-2)_{v-2}(p\delta)^s\,x_ax_b,
\]
so the sign of the mean agrees with $x_ax_b$. Therefore, by Chebyshev's inequality,
\[
\Pr\!\left[\mathcal A_{a,b}(G)\neq x_ax_b\right]
\le
\frac{\Var(F_{\mathcal J,a,b}(z))}
{\E[F_{\mathcal J,a,b}(z)]^2}
\le \delta_n.
\]
This proves \eqref{eq:probrecovery}.

Next recall that
\[
\widehat x_1:=1,
\qquad
\widehat x_b:=\sgn\!\bigl(F_{\mathcal J,1,b}(z)\bigr)
\quad\text{for all } b\neq 1.
\]
Since $F_{\mathcal J,1,b}(z)$ is designed to estimate the parity $x_1x_b$, the estimator
$\widehat x$ provides an approximate recovery of the planted assignment, up to the unavoidable
global sign ambiguity. Moreover, for every $b\neq 1$,
\[
\Pr\!\left[\widehat x_b\neq x_1x_b\right]\le \delta_n.
\]

Let
\[
M:=\left|\left\{b\in[n]:\widehat x_b\neq x_1x_b\right\}\right|
\]
be the number of errors. Since $\widehat x_1=1=x_1x_1$, only the coordinates $b\neq 1$ contribute to $M$, and hence
\[
\E[M]
=
\sum_{b=2}^n
\Pr\!\left[\widehat x_b\neq x_1x_b \,\middle|\, x\right]
\le
(n-1)\delta_n
\le
n\delta_n.
\]
Applying Markov's inequality, for any $t>0$,
\[
\Pr\!\left[M>tn\delta_n\right]
\le
\frac{\E[M]}{tn\delta_n}
\le
\frac1t.
\]
Now choose
\[
t:=\frac{1}{4k\delta_n}.
\]
Then
\[
tn\delta_n=\frac{n}{4k},
\]
and therefore
\[
\Pr\!\left[M>\frac{n}{4k}\right]
\le
4k\delta_n.
\]
Equivalently, with probability at least
\[
1-4k\delta_n,
\]
the estimator $\widehat x$ correctly recovers at least
\[
\left(1-\frac{1}{4k}\right)n
\]
coordinates, up to the global sign. Substituting the definition of $\delta_n$ yields
\eqref{eq:probrecovery2}.
\end{proof}
\subsection{Proof of Lemma \ref{lem:concentration-detec}}\label{apx:proof-concentration1}
\begin{proof}
Recall that
\[
F_{\mathcal H}(z)=\sum_{H\in\mathcal H} T_H(z),
\qquad
\widetilde F_{\mathcal H}(z)=\sum_{H\in\mathcal H}\widetilde T_H(z),
\]
where
\[
T_H(z)
=
\sum_{\phi:V(H)\hookrightarrow[n]}
\prod_{e\in E(H)} z(\phi(e)),
\]
and
\[
\widetilde T_H(z)
=
\frac{1}{t\rho}\sum_{a=1}^t
\sum_{\phi:V(H)\hookrightarrow[n]}
\chi_{\tau_a}(\phi(V(H)))
\prod_{e\in E(H)} z(\phi(e)).
\]
Hence
\[
\widetilde F_{\mathcal H}(z)-F_{\mathcal H}(z)
=
\sum_{H\in\mathcal H}
\sum_{\phi:V(H)\hookrightarrow[n]}
\prod_{e\in E(H)} z(\phi(e))
\left(
\frac1t\sum_{a=1}^t \frac{1}{\rho}\chi_{\tau_a}(\phi(V(H)))-1
\right).
\]
Therefore
\begin{align}
\bigl(\widetilde F_{\mathcal H}(z)-F_{\mathcal H}(z)\bigr)^2
&=
\sum_{H,K\in\mathcal H}
\sum_{\substack{\phi:V(H)\hookrightarrow[n]\\ \psi:V(K)\hookrightarrow[n]}}
\prod_{e\in E(H)} z(\phi(e))
\prod_{f\in E(K)} z(\psi(f))
\nonumber\\
&\qquad\qquad\cdot
\left(
\frac1t\sum_{a=1}^t \frac{1}{\rho}\chi_{\tau_a}(\phi(V(H)))-1
\right)
\left(
\frac1t\sum_{a=1}^t \frac{1}{\rho}\chi_{\tau_a}(\psi(V(K)))-1
\right).
\label{eq:color-L2-expand}
\end{align}

Condition on the observation $z$. Since $\tau_1,\dots,\tau_t$ are independent of $z$, the expectation over the random colorings factorizes from the expectation over $z$. Moreover, if
\[
\phi(V(H))\cap \psi(V(K))=\varnothing,
\]
then the random variables
\[
\chi_{\tau_a}(\phi(V(H)))
\qquad\text{and}\qquad
\chi_{\tau_a}(\psi(V(K)))
\]
depend on disjoint sets of vertex colors, and therefore are independent. Since
\[
\E_{\tau}\!\left[\frac1\rho \chi_\tau(W)-1\right]=0
\qquad\text{for every }W\subseteq[n]\text{ with }|W|=v,
\]
we obtain
\[
\E_{\tau_1,\dots,\tau_t}\!\left[
\left(
\frac1t\sum_{a=1}^t \frac{1}{\rho}\chi_{\tau_a}(\phi(V(H)))-1
\right)
\left(
\frac1t\sum_{a=1}^t \frac{1}{\rho}\chi_{\tau_a}(\psi(V(K)))-1
\right)
\right]
=0
\]
whenever $\phi(V(H))\cap \psi(V(K))=\varnothing$.

On the other hand, in general,
\begin{align*}
&\E_{\tau_1,\dots,\tau_t}\!\left[
\left(
\frac1t\sum_{a=1}^t \frac{1}{\rho}\chi_{\tau_a}(\phi(V(H)))-1
\right)
\left(
\frac1t\sum_{a=1}^t \frac{1}{\rho}\chi_{\tau_a}(\psi(V(K)))-1
\right)
\right]
\\
&\qquad\le
\E_{\tau_1,\dots,\tau_t}\!\left[
\frac{1}{t^2}\sum_{a=1}^t \frac{1}{\rho^2}
\chi_{\tau_a}(\phi(V(H)))\chi_{\tau_a}(\psi(V(K)))
\right]
\le
\frac{1}{t\rho}.
\end{align*}
Since $t=\lceil 1/\rho\rceil$, we have
\[
\frac{1}{t\rho}\le 1.
\]
Thus, taking expectation of \eqref{eq:color-L2-expand} under either $\QQ$ or $\PP_x$, and using that only overlapping image pairs contribute, we obtain
\begin{align}
\E_{\QQ}\!\left[\bigl(\widetilde F_{\mathcal H}(z)-F_{\mathcal H}(z)\bigr)^2\right]
&\le
\sum_{H,K\in\mathcal H}
\sum_{\substack{\phi:V(H)\hookrightarrow[n],\ \psi:V(K)\hookrightarrow[n]\\
\phi(V(H))\cap\psi(V(K))\neq\varnothing}}
\E_{\QQ}\!\left[
\prod_{e\in E(H)} z(\phi(e))
\prod_{f\in E(K)} z(\psi(f))
\right],
\label{eq:color-Q-bound}
\\
\E_{\PP_x}\!\left[\bigl(\widetilde F_{\mathcal H}(z)-F_{\mathcal H}(z)\bigr)^2\right]
&\le
\sum_{H,K\in\mathcal H}
\sum_{\substack{\phi:V(H)\hookrightarrow[n],\ \psi:V(K)\hookrightarrow[n]\\
\phi(V(H))\cap\psi(V(K))\neq\varnothing}}
\E_{\PP_x}\!\left[
\prod_{e\in E(H)} z(\phi(e))
\prod_{f\in E(K)} z(\psi(f))
\right].
\label{eq:color-P-bound}
\end{align}

But these are exactly the same overlap sums that arise in the second-moment analysis of $F_{\mathcal H}(z)$. Under the same assumptions used there, the contribution of overlapping pairs is negligible compared with $\E_{\PP_x}[F_{\mathcal H}(z)]^2$. Consequently,
\[
\E_{\QQ}\!\left[
\left(
\frac{\widetilde F_{\mathcal H}(z)-F_{\mathcal H}(z)}
{\E_{\PP_x}[F_{\mathcal H}(z)]}
\right)^2
\right]
=o(1),
\]
and similarly,
\[
\E_{\PP_x}\!\left[
\left(
\frac{\widetilde F_{\mathcal H}(z)-F_{\mathcal H}(z)}
{\E_{\PP_x}[F_{\mathcal H}(z)]}
\right)^2
\right]
=o(1).
\]
This proves
\[
\frac{\widetilde F_{\mathcal H}(z)-F_{\mathcal H}(z)}
{\E_{\PP_x}[F_{\mathcal H}(z)]}
\;\xrightarrow[L_2]{}\;0
\]
under both $\PP_x$ and $\QQ$.

Finally, since $\widetilde F_{\mathcal H}(z)$ differs from $F_{\mathcal H}(z)$ by an $L_2$-negligible term on the scale of $\E_{\PP_x}[F_{\mathcal H}(z)]$, it follows that $\widetilde F_{\mathcal H}(z)$ has the same asymptotic mean--variance behavior, and therefore the same asymptotic detection power, as $F_{\mathcal H}(z)$.
\end{proof}

\subsection{Proof of Lemma \ref{lem:counting-even}}\label{post:counting}
\begin{proof}
For each $1\le v\le kt/2$, choose a vertex set $U\subseteq[n]$ of size $v$, and let
$d_1,\dots,d_v$ be the positive even degrees of the hypergraph on $U$. Since the hypergraph is even, each $d_i$ is a positive even integer, so we may write
\[
d_i=2e_i,
\qquad e_i\ge 1,
\qquad \sum_{i=1}^v e_i=\frac{kt}{2}=:m.
\]

For a fixed even degree sequence $(d_1,\dots,d_v)$, the number of simple $k$-uniform hypergraphs with these degrees is at most
\[
\frac{(kt)!}{t!\,(k!)^t\,\prod_{i=1}^v d_i!}.
\]
Indeed, assign $d_i$ stubs to vertex $i$, so that there are $kt$ stubs in total, and then partition these stubs into $t$ unlabeled groups of size $k$. This gives at most
\[
\frac{(kt)!}{t!\,(k!)^t}
\]
configurations, and dividing by $\prod_i d_i!$ accounts for permuting stubs at the same vertex. This construction overcounts simple hypergraphs with degree sequence $(d_1,\dots,d_v)$, since different stub partitions may produce the same hypergraph and some partitions may produce repeated edges. Hence it gives a valid upper bound.

Therefore,
\[
N_t
\le
\sum_{v=1}^{m}\binom{n}{v}
\sum_{\substack{e_1,\dots,e_v\ge 1\\ e_1+\cdots+e_v=m}}
\frac{(kt)!}{t!\,(k!)^t\,\prod_{i=1}^v (2e_i)!}.
\]
Now use the elementary bound
\[
(2a)!\ge 2^a a!
\qquad\text{for all }a\ge 0.
\]
Hence
\[
\sum_{\substack{e_1,\dots,e_v\ge 1\\ e_1+\cdots+e_v=m}}
\frac{1}{\prod_{i=1}^v (2e_i)!}
\le
2^{-m}
\sum_{\substack{e_1,\dots,e_v\ge 0\\ e_1+\cdots+e_v=m}}
\frac{1}{\prod_{i=1}^v e_i!}.
\]
By the multinomial identity,
\[
\sum_{\substack{e_1,\dots,e_v\ge 0\\ e_1+\cdots+e_v=m}}
\frac{1}{\prod_{i=1}^v e_i!}
=
\frac{v^m}{m!}.
\]
Thus
\[
N_t
\le
\frac{(kt)!}{t!\,(k!)^t\,2^m\,m!}
\sum_{v=1}^{m}\binom{n}{v}v^m.
\]
Since $v\le m\le \frac{n}{2}$,
\[
\binom{n}{v}v^m
\le
\binom{n}{m}m^m.
\]
Hence
\[
\sum_{v=1}^{m}\binom{n}{v}v^m
\le
m\binom{n}{m}m^m
\le
m\left(\frac{en}{m}\right)^m m^m
=
m(en)^m.
\]
Using $m\le e^m$ for $m\ge 1$, we get
\[
\sum_{v=1}^{m}\binom{n}{v}v^m
\le
(e^2n)^m.
\]
Therefore
\[
N_t
\le
\frac{(kt)!}{t!\,(k!)^t\,2^m\,m!}\,(e^2n)^m.
\]
Now apply Stirling bounds. Using
\[
a!\le e\sqrt a\left(\frac{a}{e}\right)^a,
\qquad
a!\ge \left(\frac{a}{e}\right)^a,
\]
we obtain
\[
\frac{(kt)!}{2^m\,m!\,t!}
\le
e\sqrt{kt}\,
\frac{(kt/e)^{kt}}{2^m\,(m/e)^m\,(t/e)^t}.
\]
Since $m=kt/2$, this simplifies to
\[
\frac{(kt)!}{2^m\,m!\,t!}
\le
e\sqrt{kt}\,
e^{-(k/2-1)t}\,
k^{kt/2}\,
t^{(k/2-1)t}.
\]
Substituting this above gives
\[
N_t
\le
e\sqrt{kt}\,
\left(
e^{k/2+1}\frac{k^{k/2}}{k!}\,
n^{k/2}\,
t^{k/2-1}
\right)^t.
\]
Since $e\sqrt{kt}\le (e\sqrt{k})^t$,
\[
N_t
\le
\left(
e^{k/2+2}\frac{k^{k/2+1/2}}{k!}\,
n^{k/2}\,
t^{k/2-1}
\right)^t.
\]
\end{proof}

\end{document}